\documentclass[preprint]{elsarticle}
\makeatletter
\def\ps@pprintTitle{%
 \let\@oddhead\@empty
 \let\@evenhead\@empty
 \def\@oddfoot{\centerline{\thepage}}%
 \let\@evenfoot\@oddfoot}
\makeatother
\usepackage[utf8]{inputenc}
\usepackage[T1]{fontenc}
\usepackage[margin=3.5cm]{geometry}
\usepackage{lineno,hyperref}
\modulolinenumbers[5]
\usepackage[colorinlistoftodos,prependcaption,textsize=normalsize]{todonotes}
\usepackage{units}
\usepackage{subcaption}
\usepackage{booktabs}
\usepackage{makecell}
\usepackage{url}
\usepackage{amsmath,amssymb,amsfonts,amsthm}
\usepackage{mathtools}
\usepackage{bm, bbm}
\usepackage{nicefrac}
\usepackage{xcolor}
\usepackage{appendix}
\usepackage[noabbrev]{cleveref}
\usepackage[normalem]{ulem}
\newtheorem{theorem}{Theorem}

\DeclareMathOperator*{\argmax}{arg\,max}
\DeclareMathOperator*{\argmin}{arg\,min}

\newcommand{\x}{\bm{x}}
\newcommand{\z}{\bm{z}}

\journal{Journal of \LaTeX\ Templates}

\begin{document}

\begin{frontmatter}

\title{Generative Modeling of Turbulence}

\author{Claudia Drygala, Hanno Gottschalk}
\address{University of Wuppertal, School of Mathematics and Natural Sciences, IMACM \& IZMD\\
\{drygala,hanno.gottschalk\}@uni-wuppertal.de}
\author{Benjamin Winhart, Francesca di Mare}
\address{Ruhr University Bochum, Department of Mechanical Engineering, Chair of Thermal Turbomachines and Aero Engines\\\{benjamin.winhart,francesca.dimare\}@ruhr-uni-bochum.de}

\begin{abstract}
We present a mathematically well founded approach for the synthetic modeling of turbulent flows using generative adversarial networks (GAN). Based on the analysis of chaotic, deterministic systems in terms of ergodicity, we outline a mathematical proof that GAN can actually learn to sample state snapshots form the invariant measure of the chaotic system. Based on this analysis, we study a hierarchy of chaotic systems starting with the Lorenz attractor and then carry on to the modeling of turbulent flows with GAN. As training data, we use fields of velocity fluctuations obtained from large eddy simulations (LES). Two architectures are investigated in detail: we  use a deep, convolutional GAN (DCGAN)  to synthesise the turbulent flow around a cylinder. We furthermore simulate the flow around a low pressure turbine stator using the \texttt{pix2pixHD} architecture for a conditional DCGAN being conditioned on the position of a rotating wake in front of the stator. The settings of adversarial training and the effects of using specific GAN architectures are explained. We thereby show that GAN are efficient in simulating turbulence in technically challenging flow problems on the basis of a moderate amount of training data.  GAN training and inference times significantly fall short when compared with classical numerical methods, in particular LES, while still providing turbulent flows in high resolution. We furthermore analyse the statistical properties of the synthesized and LES flow fields, which agree excellently. We also show the ability of the conditional GAN to generalize over changes of geometry by generating turbulent flow fields for positions of the wake that are not included in the training data.       
\end{abstract}
\begin{keyword}
Generative adversarial networks \sep Turbulence modeling \sep Ergodicity \sep Karman vortex street \sep LPT stator
\end{keyword}
\end{frontmatter}
\nolinenumbers
\section{Introduction}
Turbulent flows are characterized by unsteadiness, chaotic-like flow states and high degree of non-linearity. The structures involved exhibit a wide range of spatial and  temporal scales, with the ratio  of largest to smallest structures scaling with the Reynolds number\cite{frisch1995turbulence}. In order to capture all scales of fluid motion directly, very fine computational meshes and time steps are required, which makes the computational effort in the case of engineering-relevant (high Reynolds numbers) problems impossible to accomplish in reasonable time despite the rapidly increasing computer performance. To circumvent this problem, closures are used, which allow to model the structures that cannot be captured by the coarser numerical meshes. However, this  advantage in computation time is paid for with a modeling error, which can be considerable depending on the chosen approach and the underlying flow case. 

Recent developments in the field of machine learning (ML), which are largely driven by increased computational power as well as the availability of exceptionally large data sets, make it possible to address this issue, whereby different approaches can be taken. 
Durbin \cite{durbin_2018} for example, provided an extensive review of currently available modelling approaches highlighting in particular the dichotomy between the need for practically viable solutions and, at the same time, the attempts to reach a universalisation or a generalisation of the representation of turbulence structures which transcends the specifics of particular applications. Whilst the popularity of scale-resolving approach is growing, the costs of industrial design and optimisation still cannot be fully compensated by the growing computing power, so that a novel approach to the improvement of low-fidelity, fast methods is urgently needed.
One obvious approach is ML-based improvement of the prediction quality of existing models, also known as ML augmented turbulence modeling. Here, one possibility is to calibrate the empirically determined constants of the respective models for the underlying use case by means of data-driven ML augmentation \cite{Cheung:2011,Edeling:2014a,Edeling:2014b,Zhang:2018a,Weatheritt:2016,Weatheritt:2017,Zhao:2020,Zhang:2018b}. Pioneering publications in this area are the works of Ling et al. \cite{Ling:2016} and Jiang et al. \cite{Jiang:2020} who used deep neural networks (DNNs) to determine the model constants of nonlinear algebraic eddy viscosity models and were thus able to significantly improve the prediction of anisotropic turbulence effects. Another way is the correction of existing models with the help of additional source terms, which were successfully used in \cite{Parish:2016,Singh:2016a,Singh:2016b,He:2018} for the augmentation of turbulence models and in \cite{Yang:2020} for the augmentation of transition models.

A completely different approach has been pursued recently, based on the generative adversarial networks (GAN) as introduced by Goodfellow \cite{Goodfellow:2014}, which allow a hierarchical identification and abstraction of features in images by means of deep neural networks (DNN). By the fact that also in the case of turbulent flows there is a complex superposition of different structures and scales suggests that these methods are well suited for learning the physical relationships in such flows. In \cite{King:2017,King:2018} it was shown that GAN are able to generate synthesizations of 2D flow fields after they have been previously trained based on DNS data. The reproductions even fulfilled some statistical constraints of turbulent flows such as Kolmogorov's $-\,\nicefrac{5}{3}$ law and the small scale intermittency of turbulence. Using a deep unsupervised learning approach and a combination of a GAN and a recurrent neural network (RNN), Kim \& Lee \cite{Kim:2020,kim2021unsupervised} were able to generate high-resolution turbulent inlet boundary conditions at different Reynolds numbers, which show a statistical similarity to real flow fields. 

Another application of GAN is the field of super-resolution reconstruction of turbulent flows. With these methods it is possible to synthetically scale up flow fields which are low-resolution or noisy due to the measurement technique used or, in the case of numerical data, due to limited data storage capacity. \cite{Fukami:2019,Fukami:2020,Liu:2020,Deng:2019,Xie:2017,Werhahn:2019,Subramanamiam:2020}. These works assume a supervised learning approach, which means that labeled paired datasets of low-resolution and high-resolution images must be available. Here, the low-resolution data sets are usually generated by filtering the high-resolution data sets obtained, for example, from direct numerical simulations (DNS). In  many practical situations, however, such high-resolution data sets are usually not available, which to a certain extent limits the range of applicability. A more general and therefore more practical approach is the unsupervised super-resolution reconstruction method. Here, pairwise data sets are no longer necessary, as Kim \& Lee could show by successfully using an unsupervised GAN for the generation of boundary conditions for turbulent flow \cite{Kim:2020} or  synthesising stationary DNS flow fields \cite{kim2021unsupervised}. Applications of such methods would be e.g. the augmentation or denoising of experimental data sets or the derivation of subgrid-scale models for the application in the field of large-eddy simulation (LES).

In our work, we show the possibility of synthesizing turbulence structures of a similar quality as predicted by of LES with GAN trained from scratch and completely unsupervised. 
As we prove by evaluation of physical quantities, the generated turbulence patterns match the statistical properties of the original LES data quite exactly. Thus, we are able to produce realistic turbulence with help of the trained generator by only having a noise vector as input. Others than in the works discussed so far, the ultimate goal of our research efforts is to devise a structural recognition workflow for a generalised, case-independent synthetisation of turbulent structures which can be carried out independently from a specific configuration. The usage of a computer vision technique sets this work aside with respect to the published works. Moreover, we show by investigation of conditional GAN that generators of synthetic turbulent flows can learn to cope with  changes of the geometry of the flow path, e.g.\ caused by a rotation wake. This remains true even if certain positions of the ware are not included into the training data. We also show that introducing generative learning to model turbulences finds its justification in the enormous reduction of computational time compared to LES, while maintaining the resolution. Lastly,  besides the practical aspects, we prove, using the mathematical concept of ergodicity, that learning to generate states of chaotic systems using GAN is possible.

\paragraph{Outline} 
The paper is organised as follows. In \cref{sec:methodology} we briefly summarize the concept of ergodicity, discuss the mathematical foundations behind GAN  along with the learning theory for deterministic ergodic systems. Also, a survey of modern GAN architectures is given. The hierarchy of datasets used for our experiments, ranging from the Lorenz attractor and the flow around a cylinder to a perioic wake impinging on a low-pressure turbine stator blade,  are described in \cref{sec:datasets}.  This is followed by \cref{sec:gan_configs}, where we  give details on the training of our various GAN models. In \cref{sec:results} we discuss the results of our numerical experiments, while in \cref{sec:physics_based_evaluation} we compare the statistical properties of simulated and synthesized turbulence. Finally,  in \cref{sec:conclusion} we present the conclusion and an short outlook.

\section{Methodology}
\label{sec:methodology}

In this work we apply generative adversarial networks (GAN) to generate typical states of a deterministic chaotic dynamic system. This is made mathematically precise via the notion of ergodicity \cite{ergodicity_basics}. 

\subsection{Ergodicity}
Let be $(\Omega, \mathcal{A}, \mu)$ a probability space, consisting of a state space $\Omega$, a collection $\mathcal{A}$ of events/subsets of the state space $A\subseteq \Omega$  known as $\sigma$-algebra and a probability measure on $\mathcal{A}$ that attributes the probability $\mu(A)$ to the events $A\in\mathcal{A}$. In our context, the state $\Omega$ is chosen as the phase space of a dynamic system $\varphi_t:\Omega\to\Omega$, $t\in\mathbb{R}$, that fulfills $\varphi_t\circ\varphi_s=\varphi_{s+t}$ and $\varphi_0(x_0)=x$. 

Frequently in this work, we need the concept of an image measure, i.e. the transformation of a measure by a mapping. To this purpose, let $\varphi : \Omega \to \Omega'$ be a mesurable mapping with respect to the $\sigma$-algebra $\mathcal{A}$ on $\Omega$ and a second sigma algebra $\mathcal{A}'$ on $\Omega'$, i.e. for all $A'\in\mathcal{A}'$ we have $\varphi^{-1}(A')=\{x\in\omega|\varphi(x)\in A'\}\in\mathcal{A}$. The image probability measure of $\mu$ under $\varphi$, denoted by $\varphi_*\mu$, is then defined by
\begin{align}
    \varphi_*\mu(A')=\mu(\varphi^{-1}(A')) ~~\forall A'\in\mathcal{A}'.
\end{align}
In the following, without further mention, we assume all mappings to be measurable with respect to suitable $\sigma$-algebras.

In the case considered here, $\Omega$ is a state space of a dynamic system. A dynamic system with the given state space consists of a collection of mappings $\varphi^t:\Omega\to\Omega$ that fulfill $\varphi^0=\mathrm{id}_\Omega$ and $\varphi^t\circ\varphi^s=\varphi^{s+t}$, where $\varphi^t\circ\varphi^s(x)=\varphi^t(\varphi^s(x))$.  In many cases, like ours, the state of the dynamic system $\varphi_t(x)$ at time $t\in\mathbb{R}$ is obtained as a solution mapping $\varphi_t:\Omega\to\Omega$ associated with a (discretized) ordinary or partial differential equation starting in the initial state $x\in\Omega$. E.g., $\Omega=\mathbb{R}^3$ for the case of the Lorenz attractor or $\Omega=\mathbb{R}^d$ with $d$ a large number of dimensions of the disrcetized state space of the fluid field in the case of the numerical simulation of turbulent fluids.

The probability measure $\mu$ is an invariant measure for the dynamic system defined by $\varphi^t$, if all solution mappings $\varphi^t$ are measure preserving with respect to $\mu$, i.e. $\varphi^t_*\mu=\mu$ for all $t\in \mathbb{R}$.

We next turn to the space of physical observables on $\Omega$ and define it as space $\mathcal{H}$ of all square-integrable  functions $f:\Omega\to\mathbb{R}$, i.e.
    \begin{align}
        \mathcal{H}:=L^2(\Omega, \mathcal{A}, \mu) =  \left \{ f:\Omega \rightarrow \mathbb{R}: f~{\rm measurable}, \int_\Omega |f|^2 \,\mathrm{d}\mu < \infty \right \}~.
    \end{align}

We next turn to to the notion of ergodicity, which equates the time average of a dynamic system with the ensemble average of its invariant measure.  In  mathematical notation, ergodicity of the dynamic system $\varphi_t$ with respect to the invariant measure $\mu$ is defined as
       \begin{align}
            \lim_{T\rightarrow \infty} \dfrac{1}{T} \int_{0}^{T} f\circ \varphi_t(x_0) \, \mathrm{d} t= \int_\Omega f(x)  \,\mathrm{d} \mu(x) =\mathbb{E}_ {\x \sim\mu}[f(\x)]~\forall x_0\in\Omega.
            \label{eq:mean_ergodic_theorem_cont}
        \end{align}
Neumann \cite{Neumann_met} and Birkhoff \cite{Birkhoff_proof} established quite general  conditions, under which ergodicity holds. See also \cite{ergodicity_farkas, ergodicity_arbabi} for extensive treatments of discrete and time continuous ergodic systems.

In some of our numerical experiments, we do not consider the entire state-space $\Omega$, but reduce the degrees of freedom using a  mapping $\pi:\Omega\to\Omega'$ with $\Omega'$ the reduced state space. Let $\pi_*\mu$ be the projected measure. Assuming the ergodictity of the original dynamics $\varphi_t$ with respect to $\mu$, we see that 
\begin{align}
            \lim_{T\rightarrow \infty} \dfrac{1}{T} \int_{0}^{T} f\circ \pi\circ \varphi_t(x_0) \, \mathrm{d} t= \int_{\Omega'} f(x') \, \mathrm{d} \pi_*\mu(x') ~\forall x_0\in\Omega,
            \label{eq:mean_ergodic_theorem_cont_reduced}
        \end{align}
whenever $f\circ \pi$ is square integrable with respect to $\mu$. This easily follows from the general transformation formula $\int_{\Omega'} f(x') \, \mathrm{d} \,\pi_*\mu(x')=\int_{\Omega'} f\circ \pi(x) \, \mathrm{d} \,\mu(x)$ and \eqref{eq:mean_ergodic_theorem_cont}. Hence, ergodicity remains meaningful on the reduced state space $\Omega'$, even if the dynamics $\varphi_t$ can not be consistently formulated on $\Omega'$.   

\subsection{Mathematical foundations of generative learning for ergodic systems}
\label{sec:foundations_of_gan_4_ergodic_systems}
Generative Adversarial Networks (GAN) consist of two  mappings - a generator $\phi :\Lambda\to\Omega$ and a discriminator $D: \Omega\to [0,1]$. Here $\Lambda$ is a space of latent variables endowed with a probability measure $\lambda$ that is easy to simulate, e.g. experiments uniform or Gaussian noise. The generator $\phi$ transforms the noise measure $\lambda$ to the image measure $\phi_*\lambda$. The goal of adversarial learning is, to learn a mapping $\phi$ from the feedback of the discriminator $D$, such that $D$ is not able to distinguish synthetic samples from $\phi_*\lambda$ from real samples from the target measure $\mu$. However, the discriminator $D$ is a classifier that trained to assign real data a high probability of being real and synthetic data a low probability. If $\phi$ has been so well trained, that even the best discriminator $D$ can not distinguish between samples from $\mu$ and $\phi_*\lambda$, generative learning is successful, see also Fig. \ref{fig:vanillaGAN}.

In practice, both the generator $\varphi$ and the discriminator $D$ are realized by neural networks. The feedback of $D$ to $\phi$ is transported backwards by back-propagation \cite{backprop} through the concatenated mapping $D\circ\phi$ in order to train the weights of the neural network $\phi$.  At the same time, the universal approximation property of (deep) neural networks guarantees that any mappings $\phi$ and $D$ can be represented with a given precision, provided the architecture of the networks is sufficiently wide and deep, see \cite{vanilla_gan, asatryan2020convenient, dcgan, pix2pixHD, pix2pix, CycleGAN2017, wang2018esrgan, karras2019style} for qualitative and quantitative results.

The training of GAN is organized as a two-player minimax game between $D$ and $\phi$. Mathematically, it is described by the min-max optimization problem
\begin{align}
    \min_\phi \max_D \mathcal{L}(D, \phi)
    \label{eq:vanilla_gan_optimization_problem}
\end{align}
with the loss function, also known as binary cross-entropy \cite{bce}
\begin{align}
     \mathcal{L}(D, \phi) = \mathbb{E}_{\x\sim \mu}[\log(D(\x))]+\mathbb{E}_{\z\sim \lambda}[\log(1-D(\phi(\z)))]~.
    \label{eq:vanilla_gan_loss}
\end{align}
Here, the expected value is denoted by $\mathbb{E}$, the random variable $\x$ with values in $\Omega$ follows the distribution $\mu$ of the real world data and the latent random variable $\z$ with values in $\Lambda$ follows the distribution of the noise measure $\lambda$.
As has been observed in \cite{vanilla_gan},
\begin{align}
    \label{eq:loss_JS}
    \max_{D\in\mathcal{H}_D}\mathcal{L}(D, \phi)=\mathfrak{d}_{\text{JS}}(\mu\|\phi_*\lambda)+\log(4)
\end{align}
if the maximum is taken over a sufficiently large hypothesis space $\mathcal{H}_D$ of discriminators. Here, $\mathfrak{d}_{\text{JS}}(\mu\|\phi_*\lambda)$ stands for an information theoretic pseudo distance between the invariant measure $\mu$ and the generated measure $\phi_*\lambda$ known as the Jensen-Shannon  divergence
\begin{align}
    \label{eq:Jenson-Shannon}
    \mathfrak{d}_{\text{JS}}(\mu\|\phi_*\lambda)=\mathfrak{d}_{\text{KL}}\left(\mu\left\|\frac{\phi_*\lambda+\mu}{2}\right.\right)+\mathfrak{d}_{\text{KL}}\left(\phi_*\lambda\left\|\frac{\phi_*\lambda+\mu}{2}\right.\right),
\end{align}
with $\mathfrak{d}_{\text{KL}}(\mu\|\nu)=-\mathbb{E}_{\x\sim \mu}\left[\log\left(\frac{f_\nu}{f_\mu}(\x)\right)\right]$ the Kulback-Leibler pseudo distance between the measures $\nu$ and $\mu$ with continuous probability densities $f_\mu$ and $f_\nu$, respectively. Note that $\mathfrak{d}_{\text{KL}}(\mu\|\nu)=0$ holds if and only if $f_\mu(\x)=f_\nu(\x)$ holds with $\mu$-probability one and hence $\mu=\nu$.  Consequently, also $ \mathfrak{d}_{\text{JS}}(\mu\|\phi_*\lambda)$ measures the distance between $\mu$ and $\phi_*\lambda$.

\begin{figure}[t]
    \centering
    \includegraphics[scale=0.9]{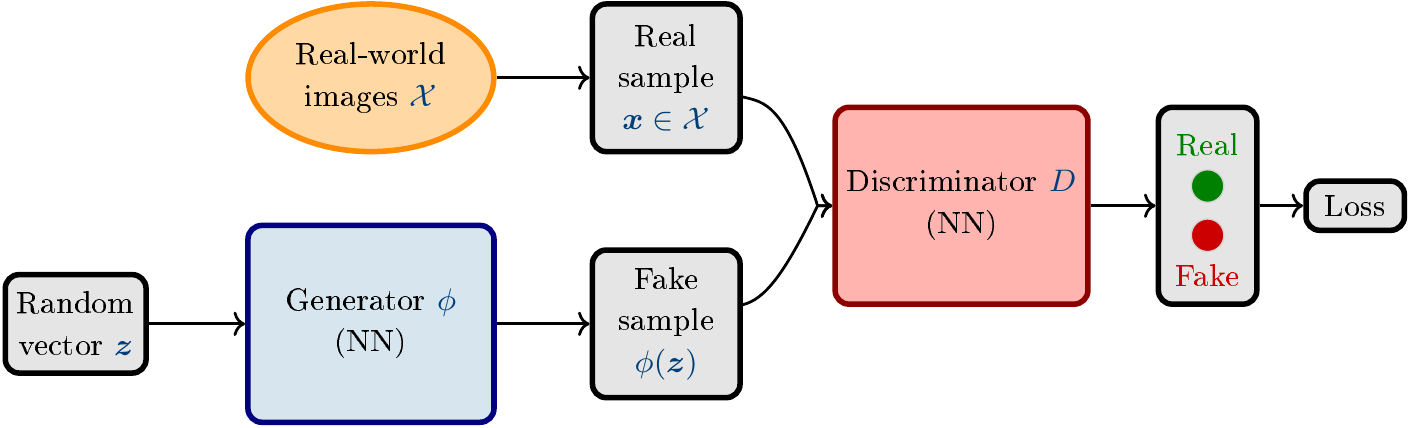}
    \caption{Architecture of the original GAN. According to: \cite{vanilla_gan}. The generator produces fake samples $\phi(\z)\sim \phi_*\lambda$ by the random vector $\z$ as its input. On the other hand the real-world data $\mathcal{X}$ representing the training data is given. The discriminator gets as input fake as well as real samples and estimates the probability that the given input sample comes from $\mathcal{X}$ than generated by $\phi$. Thus, the output of $D$ is a single scalar value per sample in the range of $[0,1]$. The feedback of the discriminator reaches the generator when the weights of the GAN framework are updated by backpropagation \cite{backprop} during the training. Since both networks are fully differentiable and trained end-to-end the whole GAN framework can be backpropagated in one go using the same loss function for $\phi$ and $D$. The optimum of the problem \eqref{eq:vanilla_gan_optimization_problem} is reached if the distribution of the real-world data is captured by the generator and the discriminator is not able to distinguish real from fake samples, so  $\phi_*\lambda=\mu$  and $D(\cdot)=\nicefrac{1}{2}$.}
    \label{fig:vanillaGAN}
\end{figure}

\subsection{Learning theory for deterministic ergodic systems}
In this work, we show that it is possible to model turbulent flows with GAN in practice. In this section we outline a proof that generative learning for deterministic ergodic systems converges in the limit of large observation time $T$. 

As described in \cref{sec:foundations_of_gan_4_ergodic_systems} $\mu$ is the unknown invariant measure encoding the statistical properties of the dynamic system $\varphi_t(x_0)$ with $x_0$ the initial state. Our goal is to sample from $\mu$ but since it is unknown, we want to learn it from the data given by the observed trajectory $\varphi_t(x_0)$. Thus, in context of generative learning a generator $\Tilde \phi$ is searched for which holds $\Tilde{\phi}_* \lambda = \mu$, where $\lambda$ is, e.g., the Lebesgue measure than corresponds to $d$-dimensional uniform noise.  

Let $\mu$ be the invariant measure of the dynamic system $\{\varphi_t\}_{t\in\mathbb{R}}$ acting  on the measurable space $([0,1]^d, \mathcal{B}([0,1]^d))$ with $\mathcal{B}([0,1]^d)$ the Borel-$\sigma$-algebra and $[0,1]^d$ the sample space of state configurations  with normalized state components in $[0,1]$. It is assumed that $d\mu(x)=f(x)d\lambda(x)$ with the continuous probability density $f(x)>0$ in the space of $k$-times differentiable $\alpha$-Hölder functions $\mathcal{C}^{k,\alpha}([0,1]^d, \mathbb{R})$ \cite{adams2003sobolev}. If this is not the case, one can easily regularize $\mu$ to achieve this. Moreover, we assume $\phi:[0,1]^d \rightarrow [0,1]^d$ also lies in the space of $k$-$\alpha$-Hölder functions $\mathcal{C}^{k,\alpha}([0,1]^d, \mathbb{R}^d)$, $k\geq 1$. By the realizability theorem of \cite{asatryan2020convenient} it follows that $\exists~ \phi_0 \in \mathcal{C}^{k,\alpha}([0,1]^d, \mathbb{R}^d)$, such that 
\begin{align}
    \phi_{0_{*}} \lambda=\mu.
    \label{eq:realizability}
\end{align}
By knowing that $\phi_{0_{*}} \lambda^{(d)}$ is realizable in the hypotheses space 
  \begin{align}
        \mathcal{H}=\{\phi \in \mathcal{C}^{k,\alpha}([0,1]^d, \mathbb{R}^d)| \lVert \phi \rVert_{\mathcal{C}^{k,\alpha}}\leq K,  \lVert \phi^{-1} \rVert_{\mathcal{C}^{k,\alpha}}\}
    \end{align}
for $K>0$ sufficiently large, our goal is to estimate $\phi_0$ by $\hat{\phi}_T \in \mathcal{H}$ based on the data given by the ergodic flow  $\bm{\varphi}_T=\{\varphi_t(x_0)\}_{T\geq t\geq 0}$.

The estimation of $\phi_0$ is performed using an empirical loss function $\hat{L}(\phi, D, \bm{\varphi}_T)$ that is designed to approximate the theoretical loss function \eqref{eq:vanilla_gan_loss} and hence minimizing the difference between the measure $\mu$ of the ergodic system and the image measure $\phi_* \lambda$ of the synthesized images. Mathematically, we search the generator
\begin{align}
    \hat{\phi}_T \in \argmin_{\phi \in \mathcal{H}} \sup_{D\in \mathcal{H}_D} \hat{\mathcal{L}}(\phi, D,  \bm{\varphi}_T)
    \label{eq:MLE}
\end{align}
 with the discriminator hypotheses space $\mathcal{H}_D$ such that an optimal choice of $D$ is feasible:
\begin{align}
    \mathcal{H}_D = \Big\{ D_{\phi, \phi'} = \frac{f_\phi}{f_\phi + f_{\phi'}} \Big| \phi, \phi' \in \mathcal{H} \Big\}~.
\end{align}
Here, $f_\phi(x)=|\det(D\phi^{-1})(x)|$ stands for the continuous probability density associated with the probability measure $\phi_*\lambda$.  We propose
    \begin{align}
    \hat{\mathcal{L}}(\phi, D, \bm{\varphi}_T) &=\dfrac{1}{T}\int_0^T\log(D(\varphi_t(x_0)))\,\mathrm{d}t + \dfrac{1}{[T]}\sum_{j=1}^{[T]} \log(1-D(\phi(z_j)))
    \label{eq:emprical_loss_ergodic_flow}
    \end{align}
as empirical loss function for the ergodic system where $[-]:\mathbb{R}\rightarrow \mathbb{Z}$ denotes the rounding function.

Apparently, in the limit $T\to\infty$ by ergodicity \eqref{eq:mean_ergodic_theorem_cont}, the first term converges to the first term in \eqref{eq:vanilla_gan_loss} whereas the second term converges almost surely by the law of large numbers. Therefore, the generator $\hat\phi$ that is learned from the empirical loss function \eqref{eq:emprical_loss_ergodic_flow} for large $T$ will approximately solve the minimax problem \eqref{eq:vanilla_gan_optimization_problem}, which by \eqref{eq:loss_JS} relates to the Jensen-Shannon distance between the estimated measure $\hat\phi_{T*}\lambda$ and the invariant measure $\mu$ of the ergodic system. In particular, we obtain the following:

\begin{theorem}
\label{theo:learning_ergodic}
Under the assumptions above it holds almost surely\footnote{w.r.t. the probability measure used for the sampling of the latent noise variables $z_j$.} that 
\begin{align}
    \label{eq:GAN_learning}
    \lim_{T\to\infty}\mathfrak{d}_{\text{JS}}(\mu\|\hat\phi_{T*}\lambda)=0.
\end{align}
\end{theorem}

\begin{proof} Here we give a sketch of the proof. For a detailed argument in a related situation, see \cite{asatryan2020convenient}.  
We introduce the following notation: $D_\phi$ is the discriminator solving $D_\phi\in \argmax \mathcal{L}(\phi, D, \bm{\varphi}_T)$ and, likewise, $\hat{D}_\phi\in \argmax \hat{\mathcal{L}}(\phi, D, \bm{\varphi}_T)$, where we suppressed the suppressed the $T$ dependence of $D$ and $\hat{D}$ to ease the notation. We obtain the estimate 
\begin{align}
        \begin{split}
        & \mathfrak{d}_{\text{JS}}(\mu|| \phi_* \lambda) = \mathcal{L}(\hat{\phi}_T, D_{\hat{\phi}_T}) - \log(4)\\
        &\leq  \hat{ \mathcal{L}}(\hat{\phi}_T, D_{\hat{\phi}_T}, \bm{\varphi}_T)+\sup_{\phi,D}\left|\mathcal{L}(\phi, D)-\hat{ \mathcal{L}}(\phi, D, \bm{\varphi}_T)\right|-\log(4)\\
        &\leq  \hat{ \mathcal{L}}(\hat{\phi}_T, \hat{D}_{\hat{\phi}_T}, \bm{\varphi}_T)+\sup_{\phi,D}\left|\mathcal{L}(\phi, D)-\hat{ \mathcal{L}}(\phi, D, \bm{\varphi}_T)\right|-\log(4)\\
        &\leq  \hat{ \mathcal{L}}(\phi_0, \hat{D}_{\hat{\phi}_T}, \bm{\varphi}_T)+\sup_{\phi,D}\left|\mathcal{L}(\phi, D)-\hat{ \mathcal{L}}(\phi, D, \bm{\varphi}_T)\right|-\log(4)\\
        &\leq  \mathcal{L}(\phi_0, \hat{D}_{\hat{\phi}_T})+2\sup_{\phi,D}\left|\mathcal{L}(\phi, D)-\hat{ \mathcal{L}}(\phi, D, \bm{\varphi}_T)\right|-\log(4)\\
        &\leq  \mathcal{L}(\phi_0, D_{\phi_0})+2\sup_{\phi,D}\left|\mathcal{L}(\phi, D)-\hat{ \mathcal{L}}(\phi, D, \bm{\varphi}_T)\right|-\log(4)\\
        &=  \mathfrak{d}_{\text{JS}}(\mu\| \phi_{0*}\lambda)+2\sup_{\phi,D}\left|\mathcal{L}(\phi, D)-\hat{ \mathcal{L}}(\phi, D, \bm{\varphi}_T)\right|\\
        &=2\sup_{\phi,D}\left|\mathcal{L}(\phi, D)-\hat{ \mathcal{L}}(\phi, D, \bm{\varphi}_T)\right|
        \label{eq:sampling_error}
        \end{split}
    \end{align}
In the first equality we used \eqref{eq:loss_JS}. In the third line, the definition of $\hat{D}_{\hat{\phi}_T}$ was used and in the fourth line we applied \eqref{eq:MLE}. In the sixth line, we used the definition of $D_{\phi_0}$. In the seventh line, we again used \eqref{eq:loss_JS} and in the final line we applied \eqref{eq:realizability}, which is possible under the given assumptions as proven in \cite{asatryan2020convenient}. 

It remains for us to show that the sampling error on the right hand side of \eqref{eq:sampling_error} vanishes as $T\rightarrow \infty$. Note that we can decompose 
\begin{align}
    \begin{split}
    &\sup_{\phi,D}\left|\mathcal{L}(\phi, D)-\hat{ \mathcal{L}}(\phi, D, \bm{\varphi}_T)\right|\\
    &\leq \sup_{\phi,D}\left|\mathbb{E}_{\x\sim \mu}[\log(D(\x))]-\dfrac{1}{T}\int_0^T\log(D(\varphi_t(x_0)))\,\mathrm{d} t\right|\\
    &+\sup_{\phi,D}\left|\mathbb{E}_{\z\sim \lambda}[\log(1-D(\phi(\z)))]-\dfrac{1}{[T]}\sum_{j=1}^{[T]} \log(1-D(\phi(z_j)))\right|
        \label{eq:decompose}
    \end{split}
\end{align}
The second term on the right hand side vanishes by the uniform law of large numbers, as the hypothesis spaces $\mathcal{H}$ and $\mathcal{H}_D$ can be endowed with $C^{k,\alpha'}$ topologies that are, for $\alpha'<\alpha$, slightly little weaker than the $C^{k,\alpha}$-topology. Nevertheless, the hypothesis spaces under these topologies are compact, see \cite{asatryan2020convenient} for the details. Consequently, the expression in the first term vanishes by the standard uniform law of large numbers, see e.g. \cite{ferguson2017course}. 

For the first term, we have already seen that ergodicity implies that the expressions in the absolute value by ergodicity vanish in the limit $T\to\infty$. Also, with respect to the aforementioned $C^{k,\alpha'}$-topologies the hypothesis spaces are compact. Last, it is easy to see that $\dfrac{1}{T}\int_0^T\log(D(\varphi_t(x_0)))\,\mathrm{d} t$ is equicontinuous in $D$ wrt.\ this topology (as $D(x)$ is uniformly lower bounded away from zero in $\mathcal{H}_D$). As for equicontinuous functions, pointwise convergence implies uniform convergence, the first term on the right hand side vanishes as well in the limit $T\to\infty$.  
\end{proof}

We note that in practice, the Hölder generators $\phi$ and discriminators $D$ are replaced by deep neural networks. As such networks possess the universal approximation property, see e.g. \cite{yarotsky2017error}, one can approximate the Hölder functions to arbitrary precision. Secondly, instead of solving the integral in \eqref{eq:emprical_loss_ergodic_flow} to compute the loss function, one uses a monte carlo approximation by sampling from the trajectory $\bm{\varphi}_T$. Theorem \ref{theo:learning_ergodic} remains valid under this replacement, as one can see from one further application of the uniform law of large numbers. 

Note however that these theoretical results do not guarantee the success of the numerical experiments. This is mostly due to the fact that the optimization problem \eqref{eq:MLE} is highly non-convex and can not be solved exactly, as e.g. for neural nets this problem is NP-hard \cite{shalev2014understanding}.  In practice, one rather finds sufficiently good local minima instead of a global optimum. Also, practical issues occur with the choice of the capacity and other elements of architecture of the neural networks.    
\subsection{Advanced GAN frameworks}
\label{ssec:GAN_frameworks}
After the introduction of the original GAN framework by Goodfellow \cref{fig:vanillaGAN},  it became apparent that GAN are powerful models which can be applied to a wide variety of tasks by modifying or extending the architecture \cite{gan_survey}. In this work three of these modified frameworks are investigated.

\paragraph{Wasserstein GAN (WGAN)}
The Wasserstein GAN differs from the original GAN mainly in the change of the loss function and thus also in the change of the optimization problem \cite{wgan}. For the WGAN framework the goal is not to minimize the Jensen Shannon divergence but the Wasserstein distance expressed by the Kantorovich-Rubinstein duality
    \begin{align}
        \mathcal{W}(\mu, \phi_*\lambda) =\dfrac{1}{K}\sup_{\lVert \psi \rVert_L \leq K} \Big(\mathbb{E}_{\x\sim \mu}[\psi(\x)]-\mathbb{E}_{\z\sim \lambda}[\psi(\phi(\z))]\Big)
        \label{eq:wgan_distance}
    \end{align}
with the supremum over all $K$-Lipschitz functions $\psi:\mathcal{C}\rightarrow \mathbb{R}$ and $\mathcal{C}$ a compact metric set.
Under the satisfaction of certain conditions the authors of \cite{wgan} showed that the optimization problem 
  \begin{align}
        \max_{\lVert \psi \rVert \leq 1} \Big(\mathbb{E}_{\x\sim \mu}[\psi(\x)]-\mathbb{E}_{\z\sim \lambda}[\psi(\phi(\z))]\Big)
    \label{eq:wgan_optimization_problem}
    \end{align}
    has a solution for $K=1$ and that the gradient of \eqref{eq:wgan_distance} exists.

In practice, the solution of \eqref{eq:wgan_optimization_problem} can be approximated by training a neural network $\psi=\psi_w$ parameterized by the weights $w \in \mathcal{W}$ with $\mathcal{W}$ a compact space. This assumption implies that all parameterized functions $\psi_w$ are $K$-Lipschitz for some $K\geq 1$. To ensure that all weights lie in a compact space and thus the Lipschitz constraint is preserved, the weights are clipped \cite{wgan_gradient_clipping} to a certain range after each gradient update in the implementation.  
    
\paragraph{Deep Convolutional GAN (DCGAN)}
The deep convolutional GAN has the same base architecture as shown in \cref{fig:vanillaGAN}, but the generator $\phi$ and the discriminator $D$ are convolutional neural networks (CNNs) \cite{dcgan}. These kind of neural networks are especially in the field of image processing successfully applicable \cite{cnn_impact, cnn_explanation}. In order to be able to integrate CNNs into GAN the authors of \cite{dcgan} pointed out which guidelines are to follow to enable a stable training at higher resolution and with deeper architectures. 

The stability of the training is ensured by applying batch normalization \cite{batch_norm} on the output layer of $\phi$ and the input layer of $D$. To work with deeper architectures fully-connected layers \cite{fc_layer} should be avoided on top of convolutional features. Finally, the choice of the leaky rectified linear unit (LReLu) activation function \cite{activation_functions} for $D$ allows higher resolution modeling. Moreover, the generator captures faster the color space of the distribution $\mu$ by applying bounded activation functions in the last layer as the LReLu \cite{activation_functions}. Finally, mentionable is that $\phi$ and $D$ are able to learn their own spatial up- or downsampling by replacing deterministic spatial pooling layers \cite{pooling} with (fractional-) strided convolutions.

\paragraph{Conditional GAN (cGAN)}
By conditioning a GAN framework with additional information it is possible to take the control over the data production process performed by the generator $\phi$ \cite{cgan}. Thereby, additional information can be represented for example by class labels or semantic segmentation masks \cite{semSeg}. As shown in \cref{fig:cGAN} the conditioning can be realized by feeding the supplementary information $\eta$ to the discriminator $D$ and the generator $\phi$ as an extra input channel. During training, $\eta$ is sampled from a data model $\bm{\eta}\sim\nu$, where $\nu$ gives the distribution of $\eta$ in the data generation process.   This extension of the architecture leads to the modified loss function 

    \begin{align}
        \mathcal{L}_{\text{cond.}}(D, \phi) = \mathbb{E}_{{\x\sim \mu\atop\bm{\eta}\sim \nu}}[\log(D(\x|\bm{\eta}))]+\mathbb{E}_{{\z\sim \lambda\atop \bm{\eta}\sim\nu}}[\log(1-D(\phi(\z|\bm{\eta})))]~.
        \label{eq:cgan_loss}
    \end{align}
    
\begin{figure}[t]
    \centering
    \includegraphics[scale=0.8]{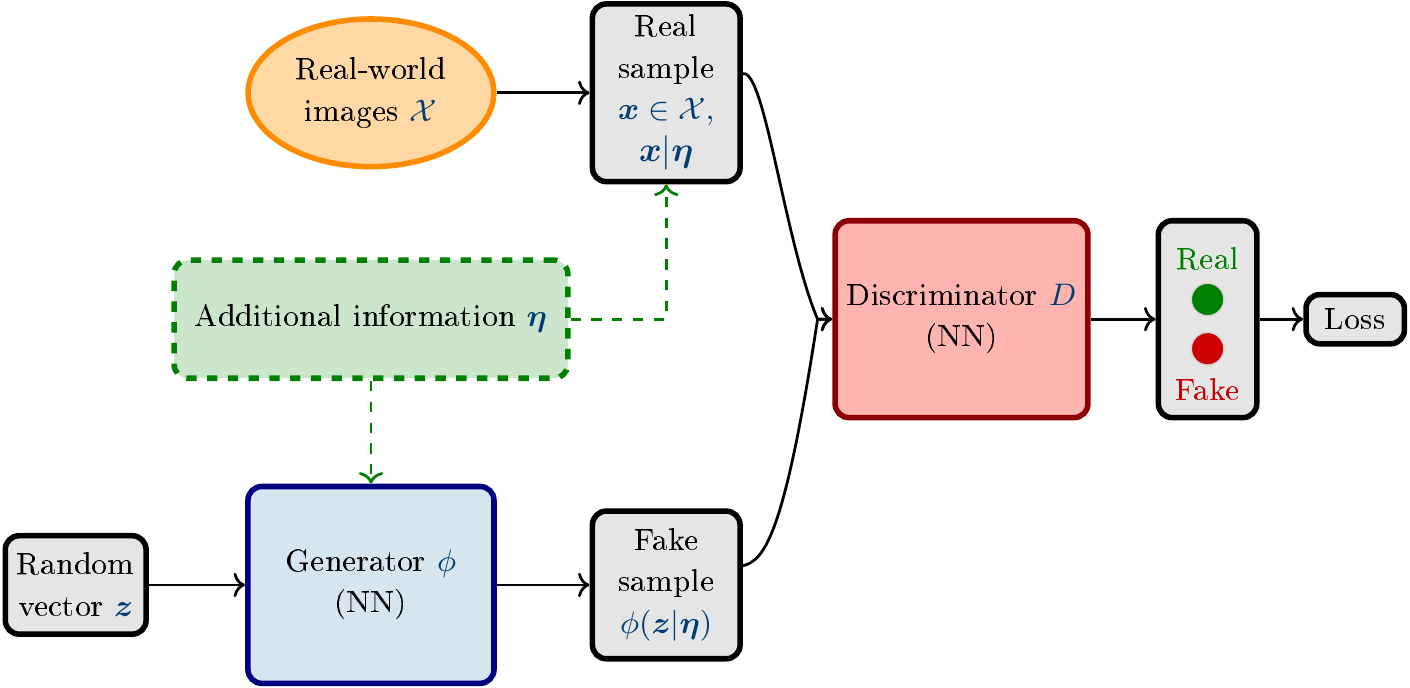}
    \caption{Architecture of a conditional GAN. According to: \cite{cgan}.}
    \label{fig:cGAN}
\end{figure}

A special form of the cGAN investigated in this work is the so called \texttt{pix2pixHD} introduced by \cite{pix2pixHD}. This conditional adversarial framework allows to generate high-resolution photo-realistic images from semantic segmentation masks. The pix2pixHD framework is based on its former version \texttt{pix2pix} \cite{pix2pix} whose optimization problem is given as
\begin{align}
    \min_\phi \max_D \mathcal{L}_{\text{cond.}}(D, \phi, \x)
    \label{eq:pix2pix_optimization_problem}
\end{align}
with $\mathcal{L}_{\text{cond.}}$ defined as in \eqref{eq:cgan_loss}.
To improve the photorealism and the resolution of the generated images the architecture was changed by introducing three innovations. 

First, a coarse-to-fine generator was implemented. For this, the generator was decomposed into the two sub-networks $\phi_1$ having the role of a global generator and $\phi_2$ as a local enhancer. By this the  global and local information can be aggregated effectively within the generator $\phi=\{\phi_1, \phi_2\}$ for the image synthesis task. 

In order for the discriminator to distinguish between generated and real high-resolution images it needs a large receptive field. Therefore, the common discriminator $D$ was replaced by three multi-scale discriminators $D_1, D_2$ and $D_3$ which have an identical network architecture, but operate at three different image scales. Hence, the optimization problem \eqref{eq:pix2pix_optimization_problem} extended to
\begin{align}
            \min_\phi \max_{D_1, D_2, D_3} \sum_{i=1}^3 \mathcal{L}_{\text{cond.}}(\phi, D_i)~.
            \label{eq:pix2pixHD_optimization_problem_multi_scales}
        \end{align}
In particular, a pyramid of images is created during the training by downsampling the input image by factor two and four. Since the discriminator operating on the coarsest scale has the largest receptive field and hence a more global view it is possible to guide the generator producing globally consistent images. Whereas, the discriminator performing on the finest scale is able to make the generator $\phi$ pay attention to finer details during the data production.

Lastly, a feature matching loss $\mathcal{L}_{FM}$ \cite{pix2pixHD} was added to \eqref{eq:pix2pixHD_optimization_problem_multi_scales}
in order to stabilize the training of the \texttt{pix2pixHD} framework. By this, the complete optimization problem is  defined as
   \begin{align}
        \min_\phi \left[ \left(\max_{D_1, D_2, D_3} \sum_{i=1}^3 \mathcal{L}_{\text{cond.}}(\phi, D_i) \right) + \gamma \sum_{i=1}^3 \mathcal{L}_{FM}(\phi, D_i) \right]
    \end{align}
    with $\gamma$ the weighting parameter for both terms.

\section{Preparation of datasets}
\label{sec:datasets}
The datasets used for generative learning are described below. We proceed from the Lorentz attactor as a simple chaotic system to LES simulations of simple and complex turbulent flows.

\subsection{Lorenz attractor}
\label{ssec:lorenz_attractor_dataset}
The Lorenz attractor is a non-periodic, non-linear and deterministic ergodic system which is given by the system of ordinary differential equations \cite{lorenz_1963} :
    \begin{align}
        \frac{dx}{dt} &= \sigma(y-x) \nonumber\\ 
        \frac{dy}{dt} &= x(\rho - z) -y
         \label{eq:lorenz_system}\\
        \frac{dz}{dt} &= xy - \beta z \nonumber
    \end{align}
By \cite{lorenz_attractor_strange_attractore_proof} it has been proven that this dynamic system is representing a strange attractor. Within this hydrodynamic system $x$ describes the rate of convection, $y$  is proportional to the temperature variation between ascending and decreasing flow and $z$ represents the distortion rate of the vertical temperature profile from linearity \cite{lorenz_1963}.

The physical parameters are given by $\sigma$ as the Prandtl number, $\rho$ as the relative Rayleigh number and $\beta$ representing the measure for the cell geometry. 
In this work we use the classic parameter values $\sigma=10$, $\rho=28$ and $\beta=\frac{8}{3}$ \cite{Kuznetsov2020}. 

The training data for the generative learning is given by the points of the attractor's trajectory within the three dimensional space computing the system \eqref{eq:lorenz_system} applying the \texttt{odeint} routine of the python package \texttt{scipy.integrate} which uses the lsoda algorithm \cite{odeint}.  In total $20,000$ data points of $200,000$ trajectories started 
from different initial points $(x_0, y_0, z_0)$ randomly sampled within the ranges $x_0\in [-40,40]$, $y_0 \in [-30,40]$ and $z_0 \in [0,50]$.

\subsection{LES}
The computational fluid dynamics (CFD) results presented in this paper form the basis for GAN training. They were generated using large-eddy simulations (LES). In this approach, the spatially filtered variant of the Navier-Stokes equations is solved, with the computational grid designed to provide a resolution of at least 80\% of the turbulent kinetic energy (TKE) of the flow. The effect of smaller turbulent structures, which are not captured by the grid, are represented using semiempirical models, the so-called subgrid scale models \cite{Hirsch:2007}. The spatial filter is thus implicitly given by the computational grid. 
The LES approach is reasonable, because it is the large vortex structures that transport the bulk of the energy \cite{Ferziger:2008} while the smaller structures can be considered to be mainly isotropic and homogeneous (not in the close vicinity of solid walls) by the assumption of local isotropy according to Kolmogorov \cite{Kolmogorov:1991}, which simplifies their modeling considerably. 
\subsection{Test-cases \& numerical setup}
Two different test cases were chosen for training of GAN, which differ in the complexity of the resulting flow field. Both simulations were performed with the commercial flow solver ANSYS Fluent which was set up to solve the incompressible variant of the spatially filtered Navier-Stokes equations. For time integration, a non-iterative time advancement scheme is used in combination with a fractional step method for pressure-velocity coupling. The advective fluxes are treated by a bounded central scheme in order to introduce as low numerical dissipation as possible to avoid unphysical dampening of small turbulent structures \cite{winhart2020large}. 
\begin{figure}[t]
    \begin{subfigure}[b]{0.35\textwidth}
	    \includegraphics[width=\textwidth]{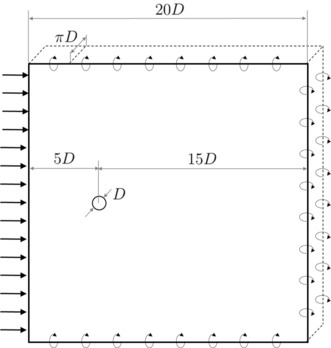}
	    \caption{Flow around a cylinder}
        \label{fig:numericalsetupA}
    \end{subfigure}
    \hfill
	\begin{subfigure}[b]{0.55\textwidth}
		\includegraphics[width=\textwidth]{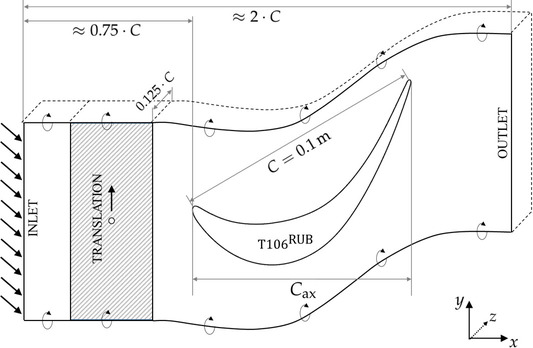}
		\caption{T106 turbine stator}
		\label{fig:numericalsetupB}
	\end{subfigure}
	\caption{Numerical domains for both investigated test cases}
\end{figure}
\subsubsection{Flow around a cylinder}
The first test case is the flow around a cylinder at Reynolds number 3900. This is a widely used test case, which has been studied in great detail in the literature both experimentally \cite{Parnaudeau:2008,Norberg:1994,Ong:1996} and numerically \cite{Parnaudeau:2008,Beaudan:1995,Kravchenko:2000}. The flow field in this case is characterized by a K\'arm\'an vortex street, that forms in the wake region of the cylinder and consists of the typical coherent vortex system, where the axis of rotation of the individual vortices is parallel to the axis of the cylinder.
A schematic representation of the numerical domain is shown in Fig. \ref{fig:numericalsetupA}. The computational grid consists of a total of 15 million cells. The time step was chosen so that the CFL number was on the order of unity, and the simulation was run for a total of 
$25,000$ time steps after initial transient effects had disappeared, which corresponds to a total physical time period of approximately $1.45$ seconds.
\subsubsection{T106 turbine stator under periodic wake impact}
\label{sssec:lpt_dataset}
The second test case is an academic low-pressure turbine (LPT) stator under periodic wake impact. In this configuration, the wakes, which are comparable to those of the cylinder test case described above, are artificially generated by means of an upstream mounted rotating bar grid.  The wakes are convected into the stator passages where deformation occurs as a consequence of the flow turning within the passage. Furthermore, a complex interaction between the wakes and the periodically detaching boundary layer takes place in the rear region of the suction side of the LPT stator, which in total makes this test case an interesting demonstrator for complex turbulent interaction phenomena. 
A schematic representation of the numerical domain is shown in Fig. \ref{fig:numericalsetupB}. The computational grid consists of a total of approx. 72 million elements. The time step was chosen so that the CFL number was on the order of unity, and the simulation was run for a total of $22,500$ time steps after initial transient effects had disappeared, which corresponds to 10 bar passing periods or approx.\ \unit[$1.43\times10^{-3}]{s}$.
\subsection{Data sets and data production}
The data sets used for training the GAN were generated by post-processing the transient LES velocity field data. In this process, grayscale images are generated via a projection mapping in the sense of \eqref{eq:mean_ergodic_theorem_cont_reduced}. In the case of the  flow around a cylinder experiment, the gray scale is showing the distribution of the absolute deviation of the local fluctuating velocity magnitude $c(\xi,t)=\sqrt{V_x(\xi,t)^2+V_y(\xi,t)^2+V_z (\xi,t)^2}$ at the location $\xi$ from its time average 
\begin{equation}
c'(\xi,t) = |c(\xi,t) - \overline{c}(\xi)|,~~\overline{c}(\xi)=\frac{1}{T}\int_0^Tc(\xi,t) \,\mathrm{d}t.
\end{equation}
 Here $V(\xi,t)=(V_x(\xi,t),V_y(\xi,t),V_z(\xi,t))$ is the velocity field of the fluid. As the moving wake determines the turbulent flow field in the case of the LPT turbine,  time averaging at a fixed point in this case does not make much sense. Therefore, a different representation of the turbulence (or projection mapping) is chosen, which simply depicts the velocity component perpendicular to the image, $w(\xi,t)$.   Figure \ref{fig:examplepictures} shows an example image for each of the two test cases examined. The gray scale for $w(\xi,t)\approx0$ is found in the upper left corner of the right panel. Negative values for $w(\xi,t)$ are shown in lighter and positive values in darker grey. 
 
 Basic parameters of the generated data sets are summarized in \cref{tab:datasetparameters}. The time step interval between two successive frames is chosen so that the respective snapshots are sufficiently far apart in time to minimize the correlation between the individual frames.
\begin{figure}[tbh]
    \begin{subfigure}[b]{0.49\textwidth}
	    \includegraphics[width=\textwidth]{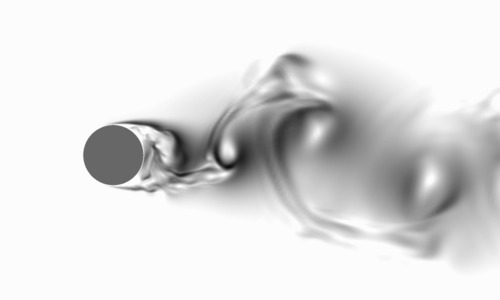}
	    \caption{Flow around a cylinder}
        \label{fig:examplecylinder}
    \end{subfigure}
    \hfill
	\begin{subfigure}[b]{0.47\textwidth}
		\includegraphics[width=\textwidth]{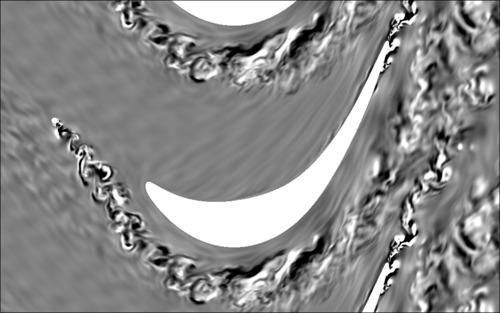}
		\caption{T106 turbine stator}
		\label{fig:exampleturbine}
	\end{subfigure}
	\caption{Example snapshots for both investigated test cases extracted from the LES.}
	\label{fig:examplepictures}
\end{figure}
\begin{table}[tbh]
	\normalsize
	\centering
	\renewcommand\cellset{\renewcommand\arraystretch{.5}%
    \setlength\extrarowheight{5pt}}
    \renewcommand\cellgape{\Gape[2pt]}
	\caption{Summary of the main data set parameters.}
	\label{tab:datasetparameters}
	\begin{tabular}{lrrrr}
		\toprule
		& \textbf{\makecell{Sampling \\ frequency}} & \textbf{\makecell{Image \\ resolution}} & \textbf{\makecell{Number \\ of files}} & \textbf{Total size} \\
		\midrule
		Cylinder & $\unit[68.9]{kHz}$ & \unit[1000 x 600]{px} & $5,000$ & $\unit[527]{MB}$ \\	
		Turbine & $\unit[40.5]{kHz}$ & \unit[1000 x 625]{px} & $2,250$ & $\unit[700]{MB}$ \\
		\bottomrule
	\end{tabular}
\end{table}
\subsection{Computational cost}
At this point, the computational effort of the simulations presented in this paper should be briefly discussed, as this is the main criterion for the applicability of such scale-resolving simulations. 

All simulations presented were performed on the in-house High-Performance Computing (HPC) cluster of the Chair of Thermal Turbomachines and Aero Engines, whose main specifications are summarized in \cref{tab:hpc}. 

In total 20 computational nodes of the \#1 partition of the HPC cluster were allocated in both runs, resulting in a total number of 560 CPU cores. 
In the case of of the flow around a cylinder, this resulted in a total computation time of about 
one day for the output run consisting of 
$25,000$ iterations, which corresponds to about 
72 core weeks.
In the case of T106 LPT stator, the calculation time was approx. 8 days for the output run consisting of $2,250$ time steps, which corresponds to 10 bar passings, i.e.\ approx.\ 640 core weeks. 

\begin{table}[tbh]
	\small
	\centering
	\renewcommand\cellset{\renewcommand\arraystretch{.5}%
    \setlength\extrarowheight{5pt}}
    \renewcommand\cellgape{\Gape[2pt]}
	\caption{Summary of the main specifications of the HPC cluster.}
	\label{tab:hpc}
	\resizebox{\textwidth}{!}{
	\begin{tabular}{cccccc}
		\toprule
		\textbf{Partition} & \textbf{\makecell{Number \\ of nodes}} & \textbf{\makecell{Cores \\ per node}} & \textbf{\makecell{CPU type}} & \textbf{\makecell{RAM}} & \textbf{Interconnect} \\
		\midrule
		\#1 & 28 & 28 & \makecell{Intel Xeon "Skylake" \\Gold 6132 @\unit[2.6]{GHz}}& \unit[96]{GB} & \makecell{Intel \\Omni-path}\\
		\#2 & 8 & 40 & \makecell{Intel Xeon Scalable \\Gold 6248 @\unit[2.5]{GHz}}& \unit[96]{GB} & \makecell{Intel \\Omni-path} \\
		\midrule
		\textbf{TOTAL} & \textbf{36} & \textbf{1104} &  & \textbf{\unit[3.4]{TB}} & \\
		\bottomrule
	\end{tabular}}
\end{table}

\section{Setup and configuration of GAN training}
\label{sec:gan_configs}
The implementations details of the training with the GAN frameworks introduced in \cref{ssec:GAN_frameworks} are summarized for the different datasets in the following. All GANs were set up and trained using the \texttt{PyTorch}  \cite{pytorch2019}.

\subsection{Lorenz attractor}
\label{ssec:gan_training_lorenz_attractor}
The Lorenz attractor was trained by a  original GAN with a discriminator consisting of four fully connected hidden layers \cite{hagan_neural_networks} with $1024, 512, 256$ and $64$ neurons. Since the attractor is a deterministic ergodic system \cite{lorenz_1963} Gaussian noise was added to the network of the discriminator as well as to the real input data to regularize the training and hence reduce overfitting \cite{asatryan2020convenient, overfitting_1995}. 
The real data representing the training data is given by the points of the attractor's trajectory within the three dimensional space as described in \cref{ssec:lorenz_attractor_dataset}.

The generator is also given by a fully connected neural network composed of three hidden layers with $256, 512$ and $1024$ neurons. Its input is given by a random vector of dimension $100\times 1$ whose elements come from the standard normal Gaussian distribution.

Both neural networks $\phi$ and $D$ apply the ReLu activation function for the input and hidden layers. The activation of the output layer of the discriminator is given by a sigmoid function and for the generator by a linear function.

The GAN framework was trained for $200,000$ epochs with a batch size of $20,000$. Hence, the trajectory consisting of $20,000$ data points was regarded during one epoch whereby the trajectory started from different randomly sampled initial points $(x_0, y_0, z_0)$ lying in the ranges $x_0\in [-40,40]$, $y_0 \in [-30,40]$ and $z_0 \in [0,50] $. 

The optimization problem was given as in \eqref{eq:vanilla_gan_optimization_problem}. To update the weights of the neural networks $\phi$ and $D$ the Adam optimizer \cite{kingma2014adam} was applied with the parameter $\beta_1=0.9$, $\beta_2=0.999$ and a learning rate of $2 \times 10^{-4}$. Here, only half of the batch size was used to update the weights of the discriminator. 

\subsection{Flow around a cylinder}
\label{ssec:gan_training_cylinder_cross_flow}
Experiments have been performed on this dataset using the  original GAN, WGAN and DCGAN framework. 
For the  original GAN and WGAN the discriminator is given by a fully connected neural network with five layers in total whereby the hidden layers consist of $1024, 512$ and $256$ neurons. The generator of both GAN frameworks also consists of five fully connected layers in total with the number of $256, 512$ and $1024$ neurons for the hidden layers.
In exception of the output layer the Leaky ReLu is applied as activation function. The last layer of the generator is activated by the hyperbolic tangent function. For the original GAN the discriminators last layer is activated by the sigmoid function and the linear activation function is used in case of the WGAN.
For the training of the DCGAN the architecture suggested by \cite{dcgan_github} was adopted. 

The three investigated GAN frameworks take images of size $k \times k$ as input. In our experiments we investigated the training with $k \in \{64, 128, 256, 512\}$. We trained all GAN for $200$ epochs with a batch size of $20$ using $5,000$ images of the dataset. For further investigations the DCGAN training was continued up to epoch $2,000$. The input vector of the generator consists of $100$ elements randomly sampled of the standard Gaussian distribution.  

For the update of the weights, the Adam optimizer is applied in case of the  original GAN and DCGAN with the parameter settings $\beta_2=0.5$ and $\beta_2=0.999$ and a learning rate of $2 \times 10^{-4}$ is used. For the WGAN the weight update is realized by the optimizer RMSProp \cite{ruder2016overview} with a learning rate of $5 \times 10^{-5}$  whereby the weights are clipped to the range $[-0,01, 0.01]$.

\subsection{T106 turbine stator under periodic wake impact}
\label{ssec:gan_training_lpt}
The DCGAN has been also trained for $2000$ epochs and $k=512$ on the whole dataset of the wake disturbed turbine stator-row with the parameter settings described as in \cref{ssec:gan_training_cylinder_cross_flow}.

Moreover, the \texttt{pix2pixHD} has been trained as second GAN framework with this dataset. As described in \cref{ssec:GAN_frameworks} the \texttt{pix2pixHD} is a conditional GAN and hence incorporates additional information to the training. Here, this supplementary information $\eta$ is given by the binary segmentation masks shown in \cref{fig:pix2pixHD_masks}. In terms of conditional GAN-learning \eqref{eq:cgan_loss}, this corresponds to a uniform distribution $\bm{\eta}\sim\nu_{\text{unif.}}$ over the $y$ coordinate of the wake. For the experiments with the \texttt{pix2pixHD} the implementation of \cite{pix2pixHD} has been used with small changes. To avoid the appearance of artifacts in the data synthesized by $\phi$ we replaced the reflection padding with a replication padding and add a replication padding to the global generator before the convolution during the downsampling procedure.

Contrary to the DCGAN framework it is possible to train the \texttt{pix2pixHD} on images of size $k \times k',~ k\neq k'$. The only important thing to take care of is that $k$ and $k'$ are divisible by $32$. For this reason, the images were resized for the training to size $k \times k' = 992 \times 624$, such that the aspect ratio has been preserved. 

Since the GAN is trained in a conditioned fashion the binary masks are also needed during the inference. For this reason, the dataset was split into a training- and test set. The training set contains the first $2000$ images of the whole dataset and the test set consists of the remaining $250$ images.   

The \texttt{pix2pixHD} has been trained for $200$ epochs with a batch size of $10$. Analogous to the DCGAN the weights were updated by the Adam optimizer with the parameter $\beta_1=0.9$, $\beta_2=0.999$ and a learning rate of $2 \times 10^{-4}$. 
\begin{figure}[t]
    \centering
    \includegraphics[width=\textwidth]{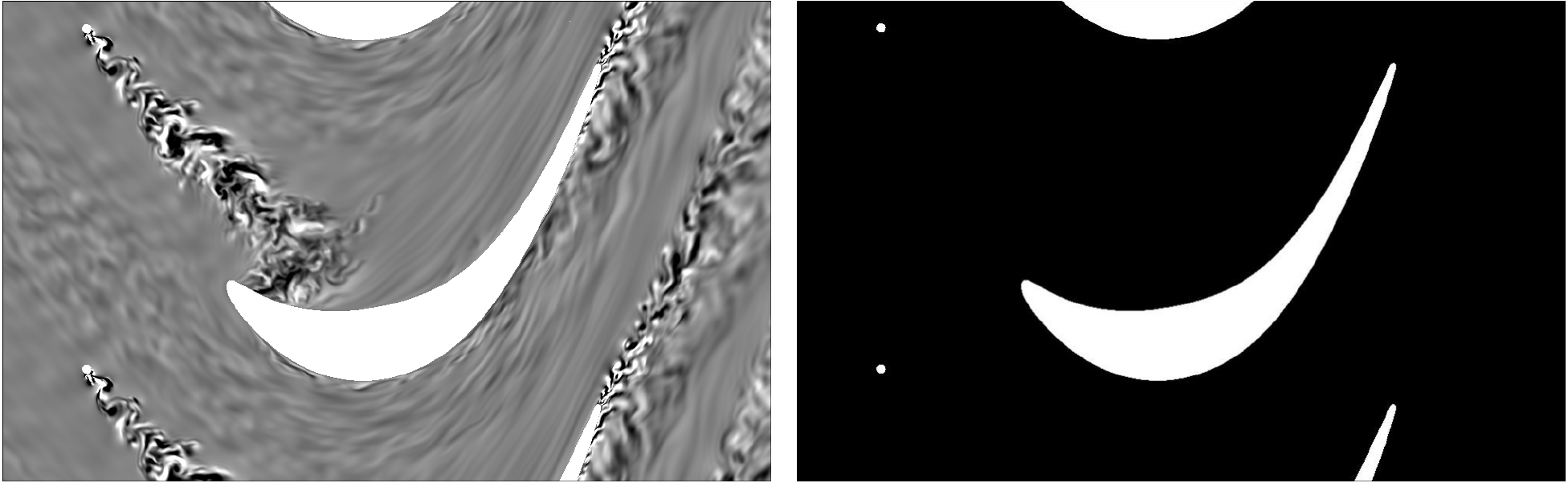}
    \caption{An image of the training set (left) and its corresponding binary segmentation mask (right).}
    \label{fig:pix2pixHD_masks}
\end{figure}

\begin{figure}[t]
    \centering
    \includegraphics[width=\textwidth]{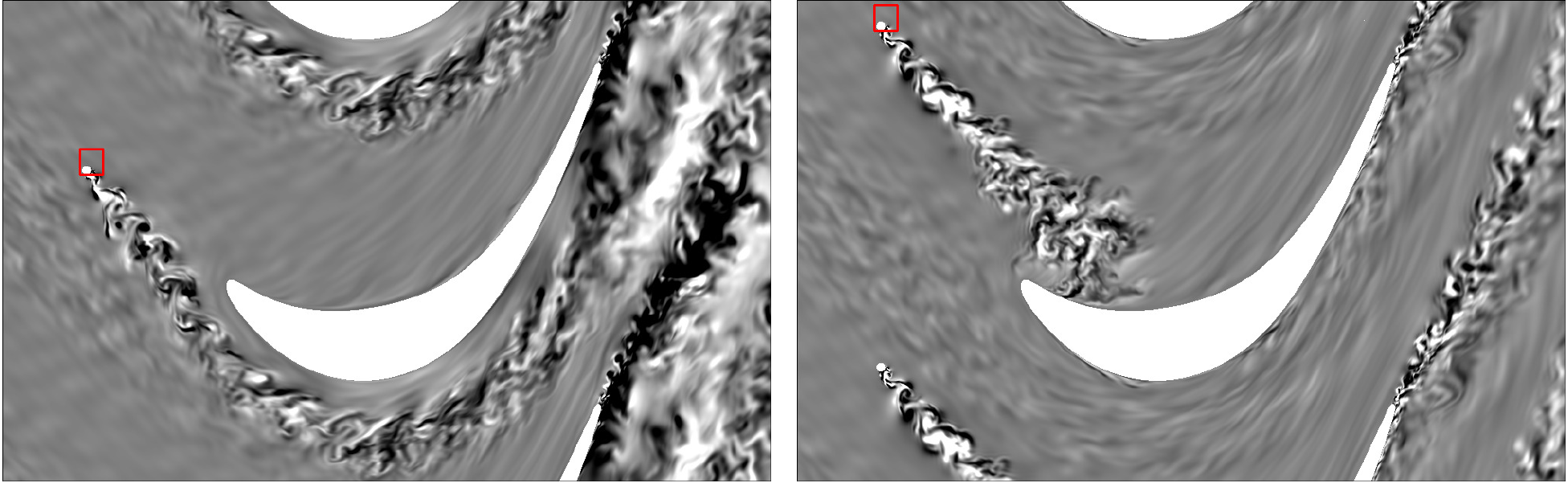}
    \caption{The regions of the wake positions (marked by red boxes) excluded in data split T1 (left) and T2 (right).}
    \label{fig:data_generalization}
\end{figure}

\subsection{Generalization over changes in geometry with the conditional GAN}
\label{ssec:data_generalization}

In order to show the ability of the conditional GAN to generalize over changes of geometry the training of the pix2pixHD has been performed using the same parameter settings as described in \cref{ssec:gan_training_lpt} but considering two further data splits. These data splits are constructed in such way that the wake positions given in the test data are not included in the training data. Thus, images have been excluded from the training data where the wake is located in a certain region. As described in \cref{sssec:lpt_dataset} the dataset of the turbine stator consists of $2,250$ images corresponding to $10$ bar passing periods. From each of these periods we excluded about $5\%$ of the images from the training representing the test data.
\Cref{fig:data_generalization} shows the regions of the wake positions which has been excluded from every period in the respective data splits. 

\section{Results of experiments}
\label{sec:results}
The results of the numerical experiments are presented and discussed in this section.
In the following, we refer to the process of applying a trained generator to the latent random vector $z$ as inference. At inference time, the latent vector also consists of $100$ elements sampled from the standard normal distribution. 

\subsection{Lorenz attractor}
As described in \cref{ssec:gan_training_lorenz_attractor} we trained a  original GAN for $200,000$ epochs in order to synthesize three dimensional data points which come from a trajectory of the Lorenz attractor that has converged towards the strange attractor. For consistency, a trajectory of $20,000$ real data points is considered at inference time as in the training. To get a better overview of the results, $500$ data points produced by the trained generator $\phi$ are shown in \cref{fig:lorenz_front}. It can be observed that the generated data points are on or close to the true trajectory of the Lorenz attractor. For the points that do not seem to lie directly on the trajectory, it has to be taken into account that the trajectory shown here is also not very dense due to the small number of data points.  
Considering randomly sampled real data points of a trajectory consisting of one million data points as it must be noted that the distribution is similar to the one of the synthesized data points. Moreover, it can be seen from the rotated \cref{fig:lorenz_side} that, apart from a few outliers, the generated data points are all located in the area of the trajectory in three-dimensional space.
\begin{figure}[htb]
    \centering
    \includegraphics[trim=0mm 0mm 0mm 15mm, clip, scale=0.65]{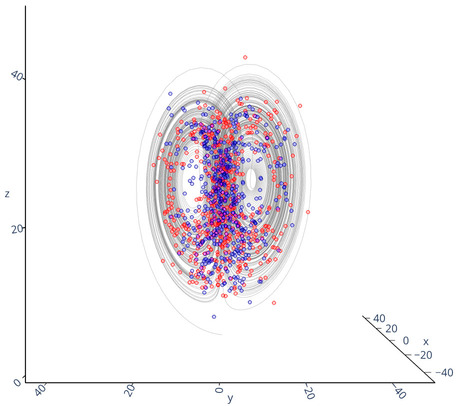}
    \caption{The $y$-$z$ plane perspective of a trajectory from the Lorenz attractor consisting of $20,000$ data points starting from the initial point $(x,y,z)=(0.1,0,0)$ (grey), $500$ synthesized data points (blue) and $500$ real data points randomly sampled from a trajectory consisting of one million data points (red).}
    \label{fig:lorenz_front}
\end{figure}

\begin{figure}[htb]
    \centering
    \includegraphics[width=\textwidth]{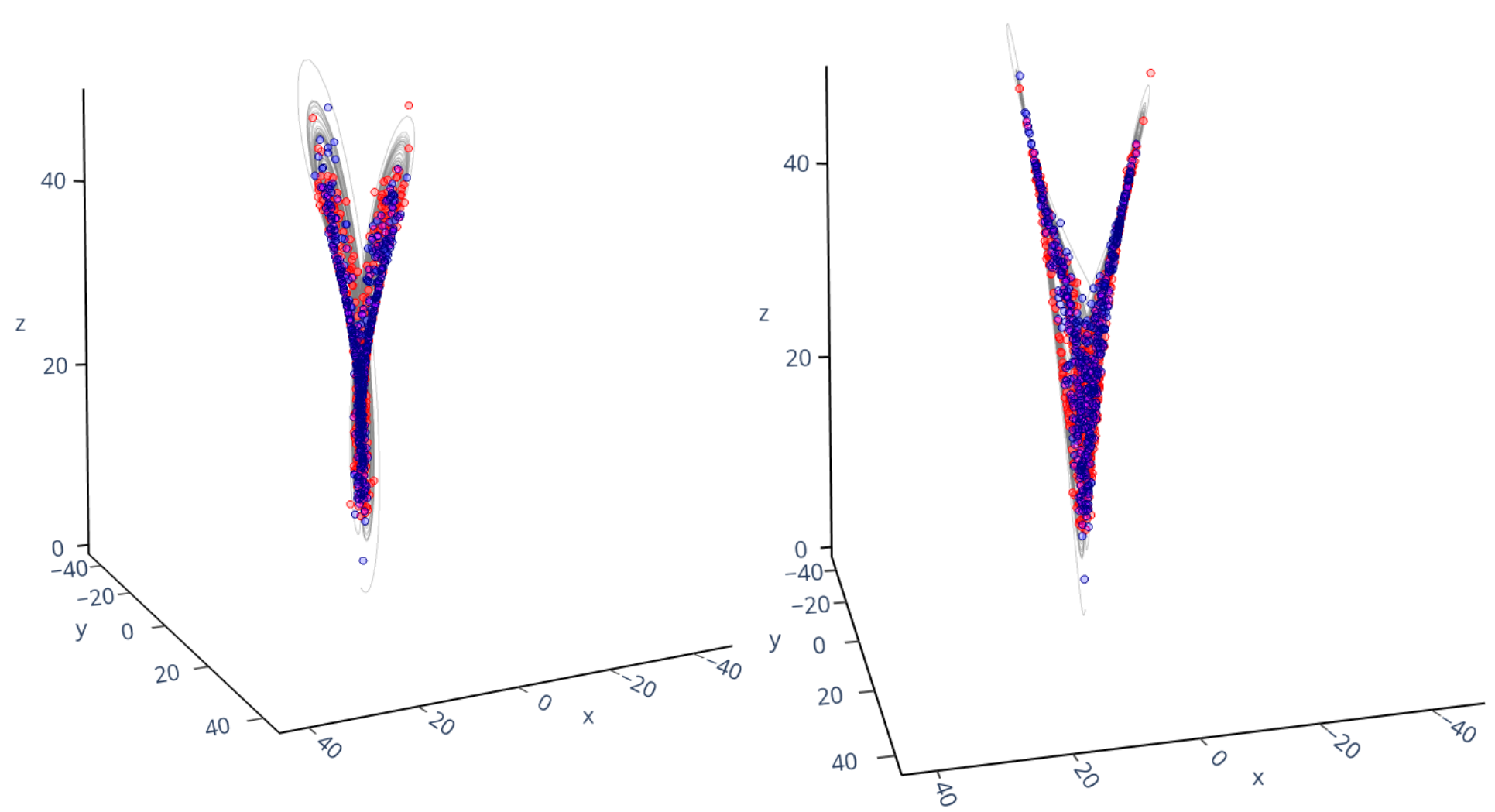}
    \caption{Rotated perspective from the trajectory of real data points (grey), synthesized data points (blue) and randomly sampled real data points (red) given in \cref{fig:lorenz_front}.}
    \label{fig:lorenz_side}
\end{figure}
\begin{figure}[t]
    \centering
    \includegraphics[width=\textwidth]{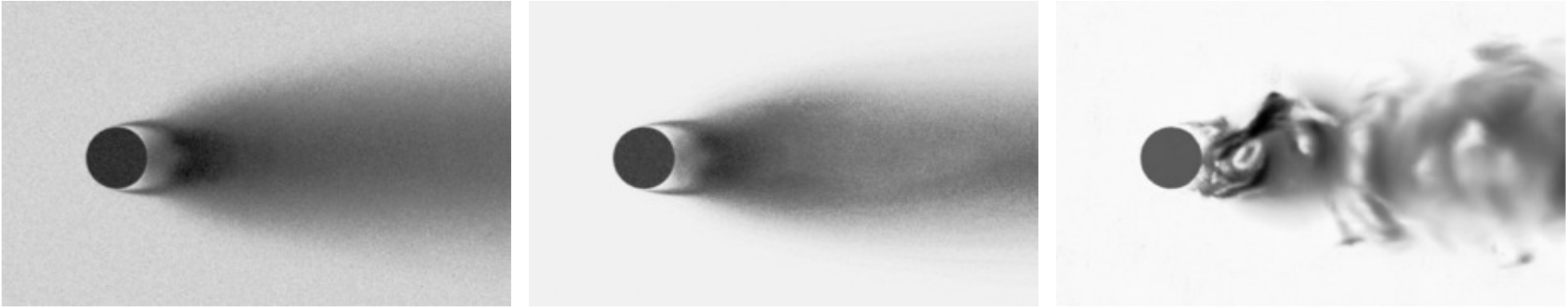}
    \caption{Comparison of images synthesized by the generator of a original GAN (left), WGAN (middle) and DCGAN (right) after $200$ epochs trained on $5,000$ images of size $k \times k$, $k=512$.}
    \label{fig:comparison_karman_various_GAN}
\end{figure}

\subsection{Flow around a cylinder}
\label{ssec:results_karman}
In order to generate the K\'arm\'an vortex street, GAN frameworks with a simpler architecture have been considered first, namely the original GAN and the WGAN. 
As to observe in \cref{fig:comparison_karman_various_GAN} the trained generators of both GAN are able to position the cylinder in the right place after $200$ epochs and that they try to synthesize the wake vortex. However, neither the original GAN nor the WGAN can capture the concrete structure of the vortex street. In addition, it is to observe that the color space has not been learned appropriate by the original GAN such that the generated images are significantly darker than the original images from the LES
(see \cref{fig:examplecylinder}). To address these issues, another GAN framework has been considered whose generator and discriminator are represented by convolutional neural networks. 
As already described in \cref{ssec:GAN_frameworks}, CNNs can be used particularly successfully in image processing. In our experiments, we also found that the DCGAN was able to capture the flow structures after $200$ epochs in contrast to the original GAN and the WGAN (see \cref{fig:comparison_karman_various_GAN}).
To increase the quality of the synthesized images the DCGAN has been further trained out to epoch $2,000$ (see \ref{sec:appendix} for the training progress). 
Based on \cref{fig:results_karman}, it can be seen that the images produced by the generator of the DCGAN hardly differ from the real images from the LES after $2,000$ epochs of training. 

Finally, it should be mentioned that the networks have been trained on images of size $k\times k$. It has been observed in our experiments that the quality of the generated images have been significantly better with increasing image resolution at inference time. Therefore, we present here the results for the training with images of size $512 \times 512$. 

\begin{figure}[t]
    \centering
    \includegraphics[width=\textwidth]{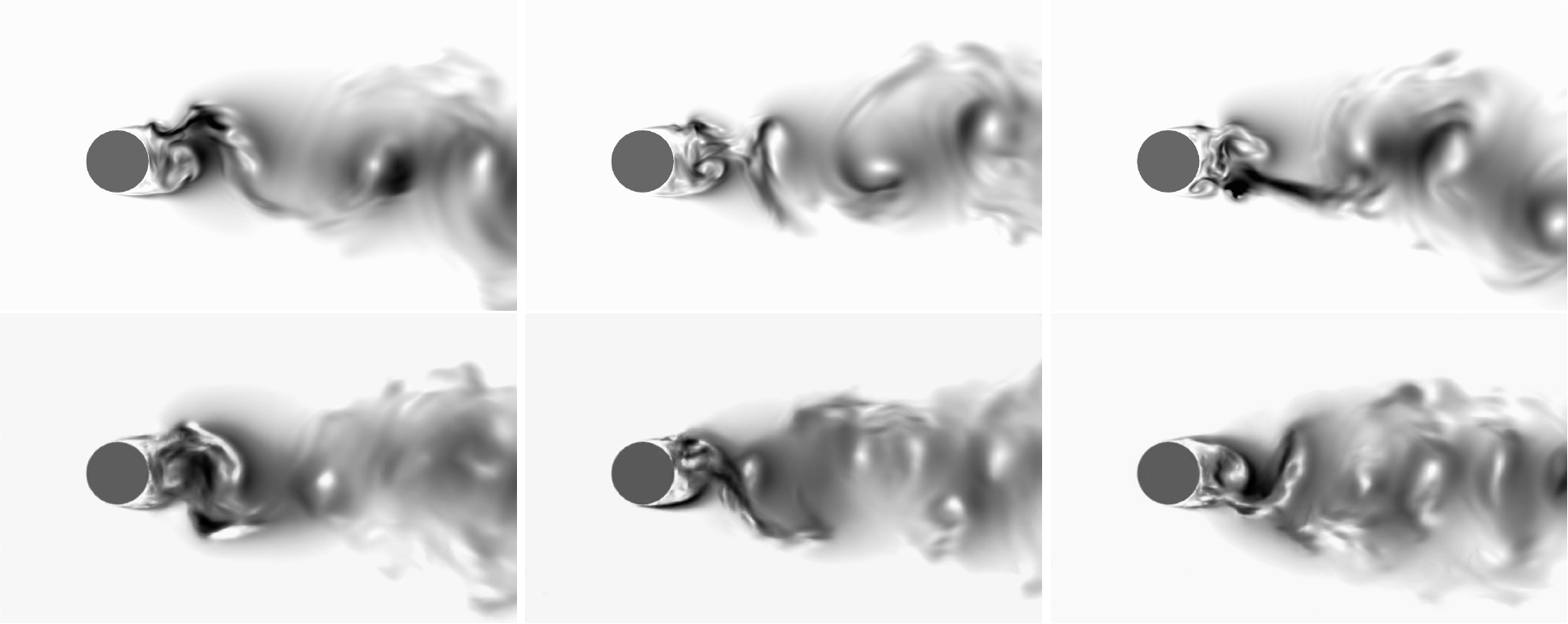}
    \caption{Comparison of images from the LES (top) and synthesized Karman vortex streets produced by the generator $\phi$ of the DCGAN trained over $2,000$ epochs (bottom).}
    \label{fig:results_karman}
\end{figure}

\subsection{T106 turbine stator under periodic wake impact}
\label{ssec:results_lpt}
Since we got impressive results from the DCGAN for the flow around a cylinder, we trained this GAN framework under the same parameter settings for the second test case. As we observe in \cref{fig:results_lpt_dcgan_fails}, the LPT stator has been correctly positioned and the structure of the vortex flows has been also reasonably captured. However, at inference time, the generator has massive problems correctly capturing the position of the cylinder as it periodically slides from bottom to top over time. Especially by direct comparison in \cref{fig:comparison_lpt_dcgan_fails} we can observe, that the structures in the background are not properly captured and the synthesized images are significantly darker than the real images of the LES. 
To address these problems of the DCGAN we considered the \texttt{pix2pixHD} as another GAN framework. In order to have control over the position of the cylinder at inference, we feed binary segmentation masks shown in \cref{fig:pix2pixHD_masks} as additional information $\eta$ to the GAN framework during training and at inference time (see \cref{ssec:GAN_frameworks}). These masks have the information about the position of the cylinder and the LPT stator. 
Moreover, we are allowed to generate high resolution images by the \texttt{pix2pixHD} framework such that the structure in the background of the images should also be preserved.

As shown in \cref{fig:results_lpt}, using the  generator from \texttt{pix2pixHD} we were able to generate images which again can be hardly qualitatively distinguished from the real image from the LES on a visual level after only $200$ epochs (see \ref{sec:appendix} for the training progress). It is also noticeable that the wake vortices do not look identical. Hence, the generator did not simply memorize the structure of the wake vortices at the respective positions and thus variation is given in the synthesized data.

\begin{figure}[t]
    \centering
    \includegraphics[width=0.95\textwidth]{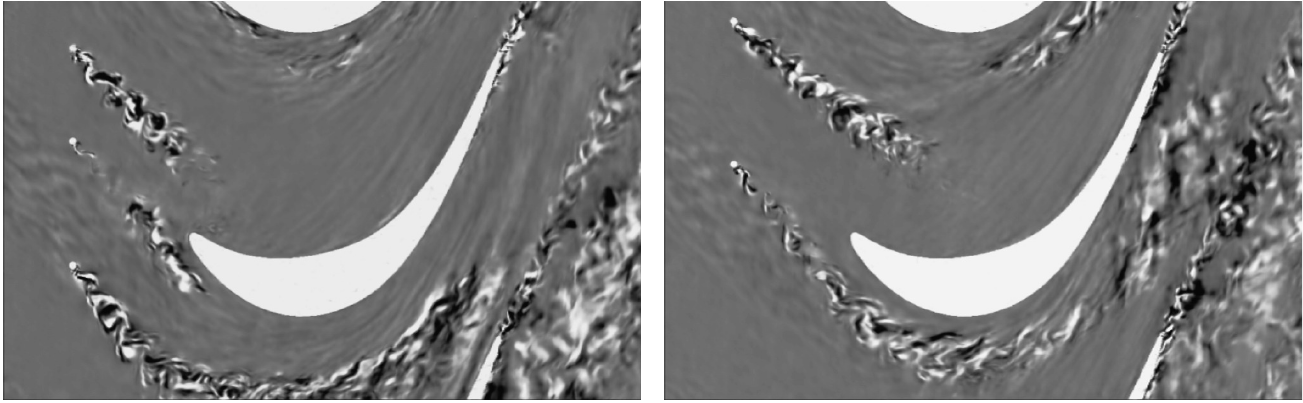}
    \caption{Examples of images synthesized by $\phi$ of the DCGAN trained on $2,250$ images for $2,000$ epochs.}
    \label{fig:results_lpt_dcgan_fails}
\end{figure}

\begin{figure}[t]
    \centering
    \includegraphics[width=0.95\textwidth]{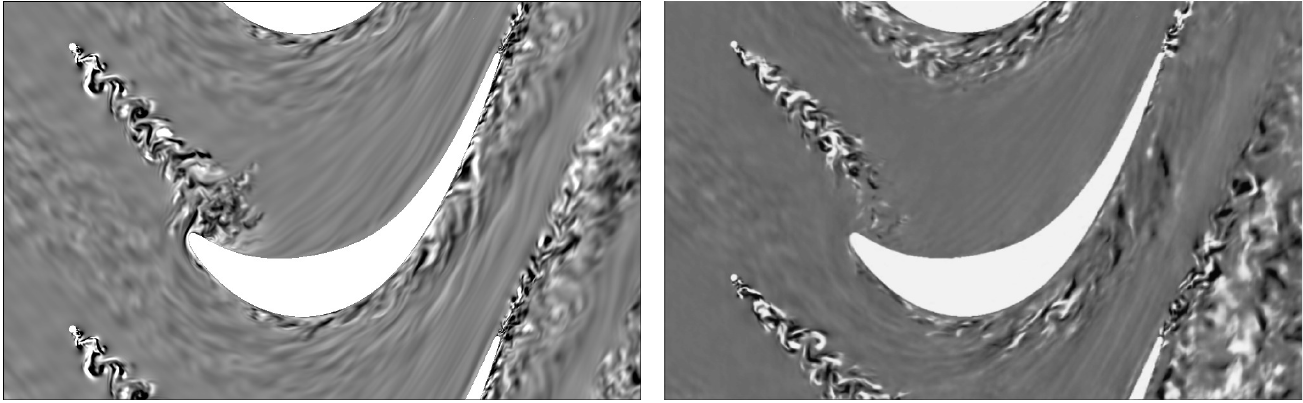}
    \caption{Comparison of a real images from the LES (left) and an images synthesized by $\phi$ of the DCGAN trained on $2,250$ images for $2,000$ epochs (right).}
    \label{fig:comparison_lpt_dcgan_fails}
    \end{figure}

\begin{figure}[!h]
    \centering
    \includegraphics[width=\textwidth]{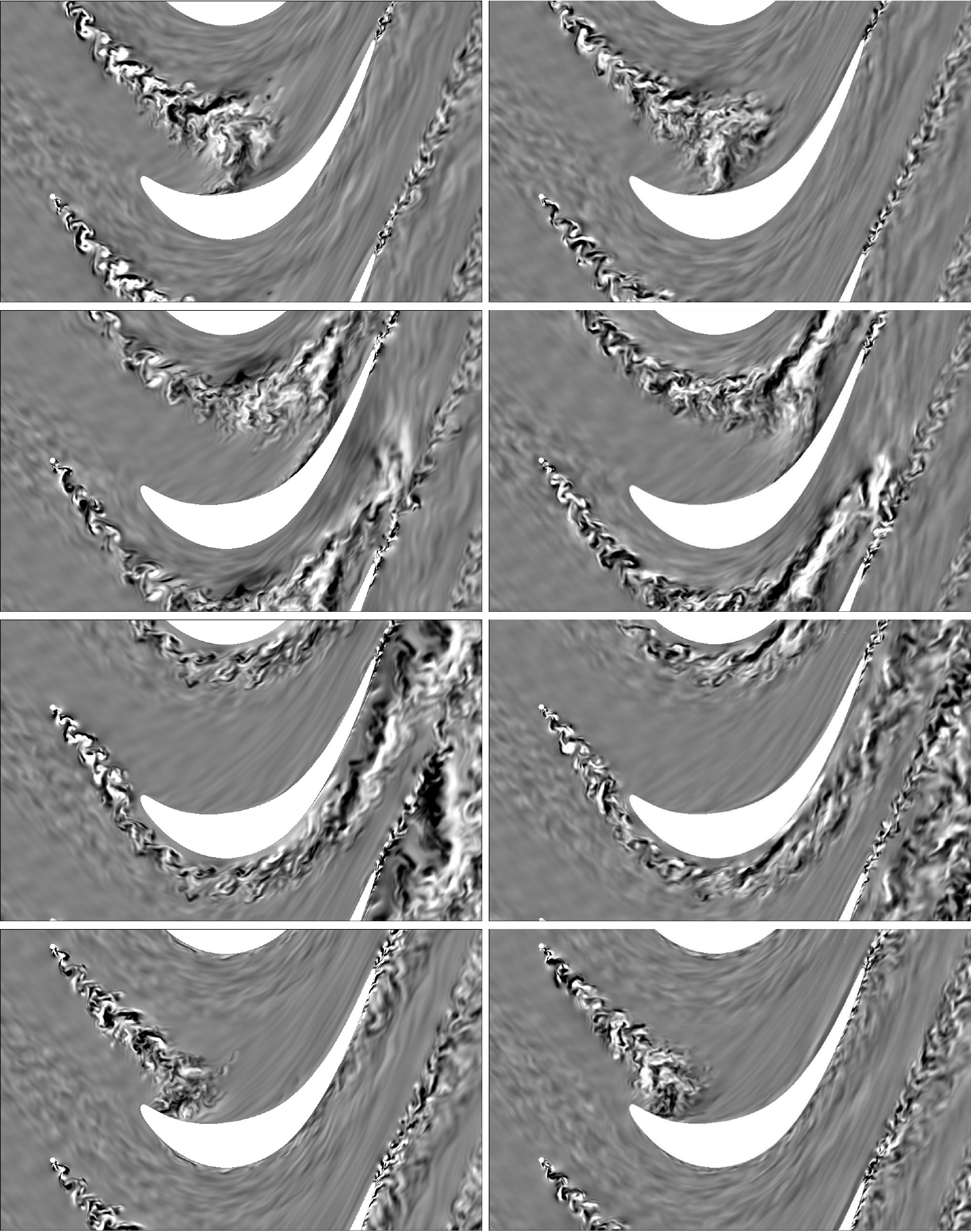}
    \caption{Comparison of images from the LES (left) and synthesized turbulences under periodic wake impact produced by the generator $\phi$ of the pix2pixHD trained over $200$ epochs (right).}
    \label{fig:results_lpt}
\end{figure}

\subsection{Generalization over changes in geometry}

As described in \cref{ssec:results_lpt} we are able to control the position of the wake at inference time by the conditional GAN framework \texttt{pix2pixHD}. After observing that visually high quality turbulent flow fields can be generated by this framework, we went a step further in additional experiments and investigated the ability of the \texttt{pix2pixHD} to generate turbulent flow fields for positions of the wake not considered in the training. Therefore, we trained the GAN framework  using the two data splits introduced in \cref{ssec:data_generalization}. At inference time, the generator was applied to binary segmentation masks with wake positions unseen during the training. In  \cref{fig:results_generalization_split1} and \cref{fig:results_generalization_split2} we observe that the generator managed this task and was able to produce turbulence of high quality for the unknown wake positions for both data splits. Furthermore, the generated wake vortices do not look the same and naturalistic variation is given in the synthesized data. 

\begin{figure}[!h]
    \centering
    \includegraphics[width=0.97\textwidth]{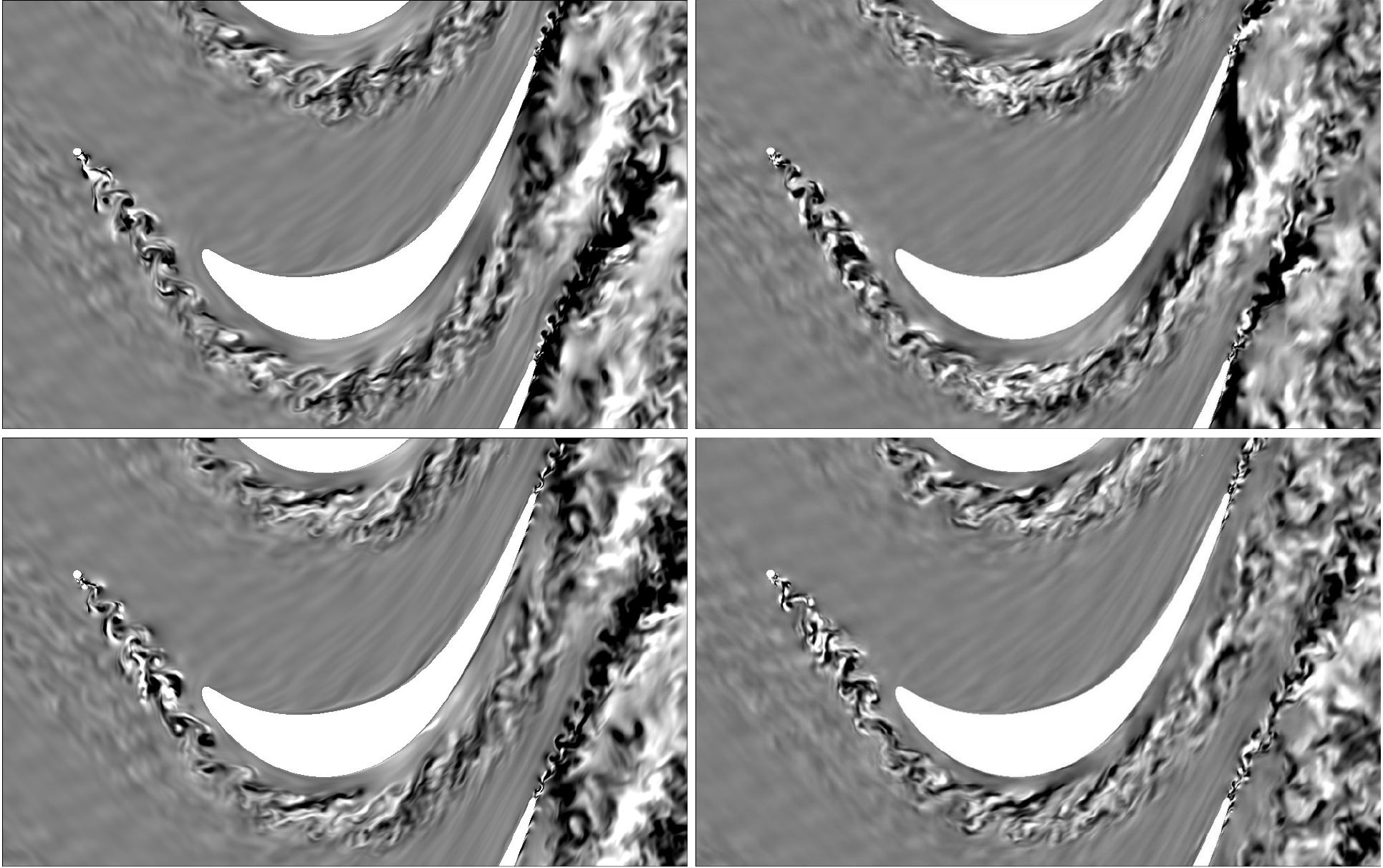}
    \caption{Comparison of images from the LES (left) and synthesized turbulences under periodic wake impact produced by the generator $\phi$ of the pix2pixHD trained over $200$ epochs (right) using data split T1. }
    \label{fig:results_generalization_split1}
\end{figure}

\begin{figure}[!h]
    \centering
    \includegraphics[width=0.97\textwidth]{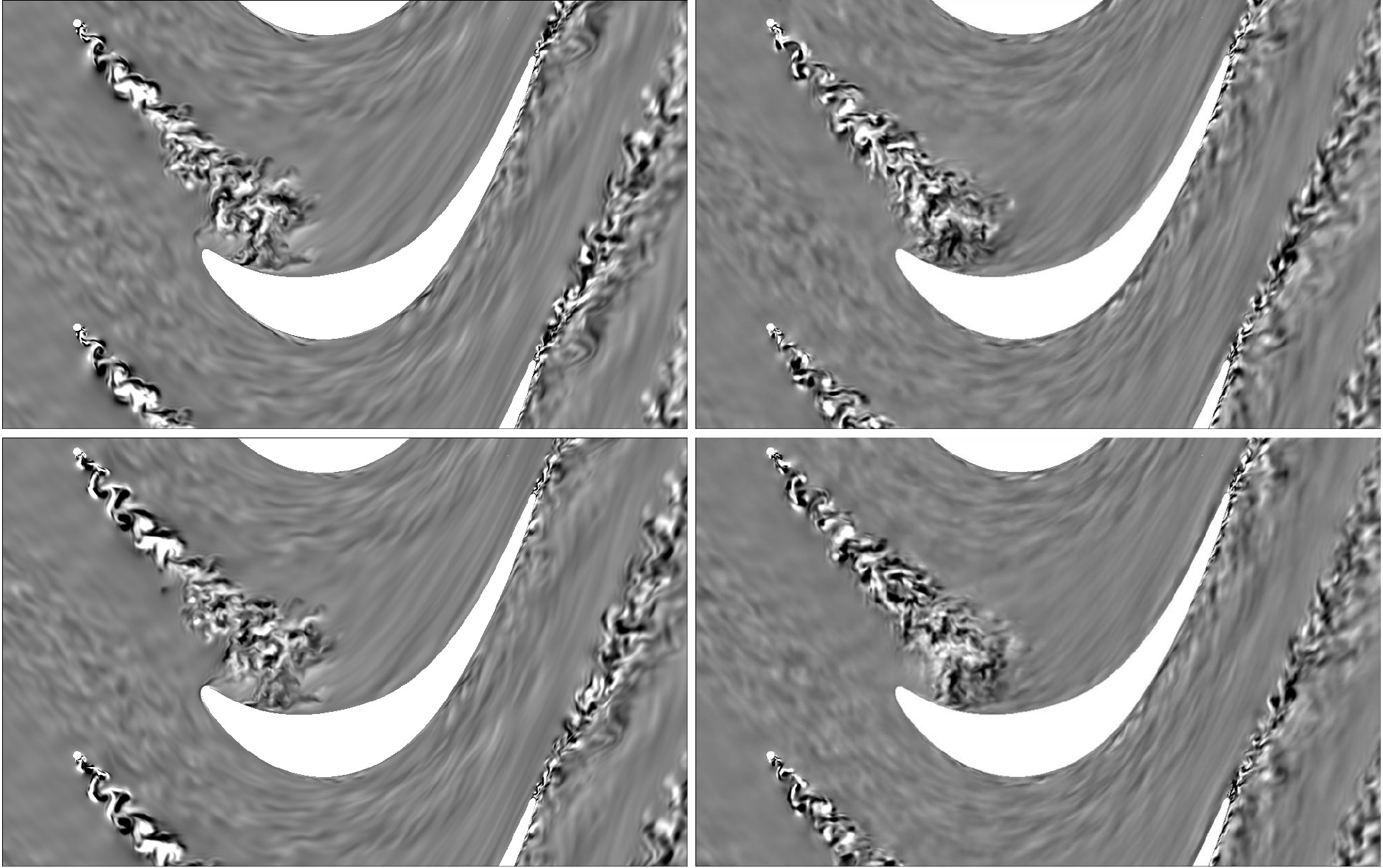}
    \caption{Comparison of images from the LES (left) and synthesized turbulences under periodic wake impact produced by the generator $\phi$ of the pix2pixHD trained over $200$ epochs (right) using data split T2. }
    \label{fig:results_generalization_split2}
\end{figure}

\subsection{Comparison of Computational Costs}
Finally, the computational costs of the training and inference performed on a GPU of type Quadro RTX 8000 with $48$ GB of the successful GAN frameworks are reported in this section. 

The training of the DCGAN with $5,000$ images of the dataset showing the flow around a cylinder has taken  $1.5$ minutes per epoch. The computational time of pure inference is given by $0.001$ seconds per image. Thus, the production of a dataset containing $5,000$ images would take with the beforehand trained generator about $5$ seconds. This leads to a tremendous amount of time saved compared to one day needed for the generation of the images by the LES.

Since the \texttt{pix2pixHD} has a much more complex architecture than the DCGAN the training of one epoch with $2,000$ images has taken $17$ minutes. However, the computational time of pure inference is also given by only $0.01$ seconds per image. Hence, the production of $2,250$ images of the LPT stator under periodic wake impact would take  about $\unit[22.5]{s}$  at inference. Thus, the saved computational time for the data production  is very significant in comparison to $8$ days for the LES. 

\section{Physics-based evaluation}
\label{sec:physics_based_evaluation}

In this section, we compare the statistical properties of GAN-generated turbulence patterns with those of the LES original. We note that all quantities evaluated here can be cast in the abstract form $\mathbb{E}_ {\x \sim\mu}[\psi(\x)]$ of \eqref{eq:mean_ergodic_theorem_cont} and \eqref{eq:mean_ergodic_theorem_cont_reduced}, respectively, with  specific evaluation functions $\psi$. Theorem \ref{theo:learning_ergodic} can be seen as a proof that the evaluation of such physical quantities on GAN-generated data in average and in the limit of large data and large network capacity will be arbitrarily close to the corresponding evaluation on the data generated by the original ergodic system. 

This statement is a consequence of the well-known fact from information theory that convergence with respect to the norm of total variation is equivalent to convergence with respect the Jensen-Shannon divergence \eqref{eq:Jenson-Shannon}.  For the convenience of the reader we provide the details in \ref{sec:appendix-B}. From Theorem \ref{theo:learning_ergodic} it thus follows that
\begin{equation}
    \label{eq:total-variation}
    \|\hat \phi_{T*}\lambda -\mu\|_{\textup{TV}}=\sup_{\|\psi\|_\infty\leq 1}\left|\mathbb{E}_{\x\sim\hat \phi_{T*}\lambda}[\psi(\x)]-\mathbb{E}_{\x\sim\mu}[\psi(\x)]\right|\to 0,  
\end{equation}
holds almost surely as $T\to\infty$, with $\|\psi\|_\infty=\sup_{x\in [0,1]^d}|\psi(x)|$. Hence, the convergence of any statistics of the flow field evaluated on GAN data converges to the corresponding statistics evaluated on LES data and this convergence is even uniform over all uniformly bounded functions $\psi$. Therefore, in contrast to \cite{kim2021unsupervised}, we do not include any loss functions that are directly related to statistical evaluation to our GAN training procedures, but only evaluate in as much our numerical experiments comply with the theoretical insight presented.

As already mentioned at the end of \cref{sec:foundations_of_gan_4_ergodic_systems}, this theoretical result relies on assumptions that  are only approximately fulfilled in practical applications. We therefore provide certain examples for physically meaningful evaluation functions $\psi$ and provide numerical tests for these. 

We start with an evaluation of the variance of the local fluctuating velocity magnitude $c(\xi,t)$ at point $\xi$ in case of the flow around a cylinder. Here the data $x$ contains the (normalized) values of $c(\xi,t)$ with $\xi(t)=\varphi_t(x(0))$ from a grid of pixels and $t$ from certain time slices described in section \ref{sec:datasets}. For $\psi(x)=x_\xi^2$ we evaluate the squared pixel value at pixel $\xi$. Note that this function is bounded on the normalized data and $\textup{Var}[c(\xi,t)]=\frac{1}{T}\int_0^Tc'(\xi,t)^2\,\textup{d} t=\mathbb{E}_{\x\sim\mu}[x_\xi^2]$ measures the strength of statistical fluctuation of velocities at the point $\xi$.

\begin{figure}[htb]
    \begin{subfigure}[b]{0.5\textwidth}
	    \includegraphics[width=\textwidth]{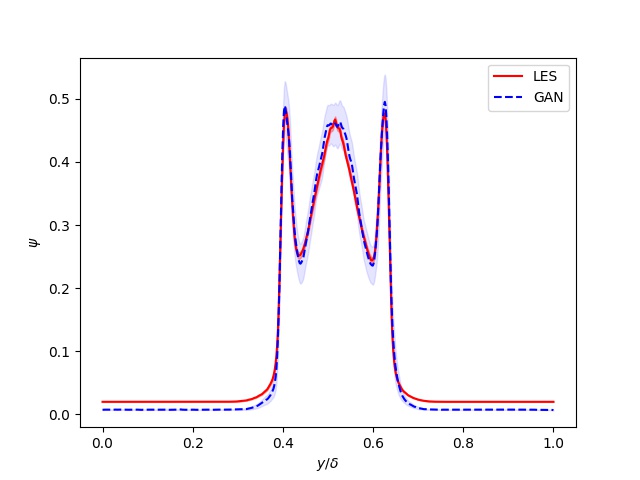}
	    \caption{Mean pixel values}
        \label{fig:eval_res_mean_small}
    \end{subfigure}%
	\begin{subfigure}[b]{0.5\textwidth}
		\includegraphics[width=\textwidth]{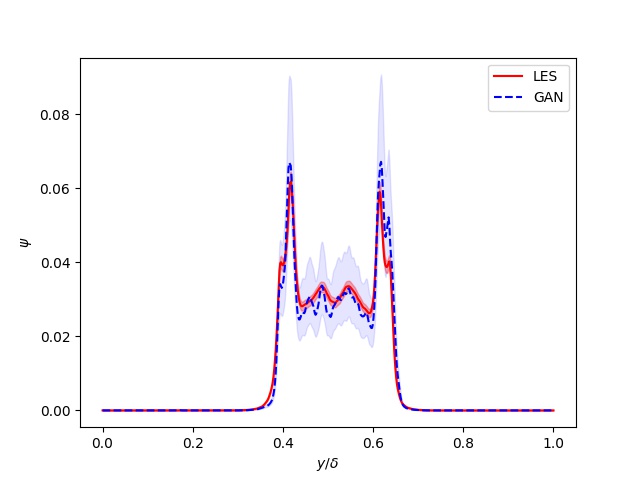}
		\caption{Statistical fluctuation of deviation from mean velocities}
		\label{fig:eval_res_var_small}
	\end{subfigure}
	\caption{Comparison of the mean pixel values (a) and the statistical fluctuation of the deviation from mean velocities (b) for the K\'arm\'an vortex street simulated by LES and generated by GAN along the $y$-axis. The blue and red shaded areas indicate the $95\%$ confidence intervals of the variance for the respective curves. Both data sets were normalized before evaluation.}
	\label{fig:evaluation_result_variance_mean_small_region}
\end{figure}

\Cref{fig:eval_res_var_small} shows the averaged variances over a small grid of $12$ pixels in $x$-direction immediately after the wake along the $y$-axis. As to observe in \cref{fig:examplecylinder} the strongest statistical fluctuation of the velocities is given immediately after the vertebral body. $\delta$ corresponds to the bandwidth of the vortex body, so that by $y/\delta$ the normalized pixel index value of the $y$-axis is defined.
For the dataset of the flow around a cylinder we realized the evaluations with $5,000$ images of the LES used in the training and $5,000$ synthesized images by the GAN. 
Although no physical quantities have been included in the GAN training we can observe that the statistical fluctuations of the velocities of the turbulence from the LES and synthesized by the DCGAN show a similar development and that the confidence intervals of the respective curves are overlapping.  

In addition to the variance we also compared the mean pixel values of the vortex streets produced by the LES and the GAN, i.e. $\psi(x)=x_\xi$.
\Cref{fig:eval_res_mean_small} shows that the GAN approximates the LES almost perfectly in the investigated region for this metric.

Extending the investigated grid and considering about $25\%$ of the pixel values after the cylinder, we can hold our observations and find that in particular the statistical fluctuations of the turbulence generated by the GAN stabilize. The results of these analysis are provided in \cref{sec:appendix-C}.

For the dataset of the turbine stator under periodic wake impact, the synthetic image $\x_\xi$ evaluated at the pixel $\xi$ synthesizes a snapshot of the $z$-component $V_z(\xi,t)$ of the velocity field at fixed time $t$. We compare the correlation of GAN-generated turbulence patterns $\x_\xi=V_z(\xi)$ with those of the LES
\begin{equation}
\label{eq:correlation-coefficient}
    \rho_{V_z, V_z}(p) = \dfrac{\textup{cov}_\mu\left[V_z(\xi(\tau)), V_z(\xi(\tau)+p)\right]}{\sigma_{V_z(\xi(\tau))} \sigma_{V_z(\xi(\tau)+p)}}.
\end{equation}
Here $\xi(\tau)$ stands for a point that is co-moving with the wake and $p$ is a vector that points against the direction of the vector connecting $\xi(\tau)$ with the wake, see Figure \ref{fig:evaluation_correlation_area}. 
Here, $\textup{cov}$ denotes the covariance, $\sigma$ the standard deviation. The co-moving pixel $\xi(\tau)$ is chosen such as it exposes a high level of variation in $V_z(\xi(\tau))$. Note that 
\begin{align}
\label{eq:covariance-vs-psi}
    \begin{split}
      &\textup{cov}_\mu\left[V_z(\xi(\tau)), V_z(\xi(\tau)+p)\right]\\
      =&\mathbb{E}_{\x \sim \mu_\tau}[\psi_{2,\xi(\tau),p}(\x)]\\
      -&\mathbb{E}_{\x \sim \mu_\tau}[\psi_{1,\xi(\tau)}(\x)]\mathbb{E}_{\x \sim \mu_\tau}[\psi_{1,\xi(\tau)+p}(\x)],
    \end{split}
\end{align}
where $\mu_\tau$ is the limiting measure of the ergodic flow conditioned to the wake position at time $\tau$, where we used evaluation functions  $\psi_{1,\xi}(\x)=\x_{\xi}$ and $\psi_{2,\xi,p}(\x)=\x_{\xi}\x_{\xi+p}$. For the representation of the standard deviation via evaluation functions we proceed analogously. From a theoretical standpoint, by application of \eqref{eq:covariance-vs-psi} and a generalization of \eqref{eq:total-variation} to the situation of the rotating wake  we can thus infer the convergence of \eqref{eq:correlation-coefficient} for $\mu$ replaced by $\hat\phi_{T*}\lambda$ to the correlation coefficient of the LES \eqref{eq:correlation-coefficient}. 

We realized the evaluation of correlation coefficients utilizing empirical correlation based on $225$ images of the test set produced by the LES and the GAN framework \texttt{pix2pixHD}. Analogously to the evaluations of the  K\'arm\'an vortex street, we investigated a certain area of the turbulent flow fields described in \cref{fig:evaluation_correlation_area}. \Cref{fig:evaluation_result_correlation} shows that the pointwise correlation of the GAN- generated turbulence coincide with those of the LES within the error bounds given by a 95\% confidence interval.

Thus, we can summarize that the turbulent flow fields synthesized by the GAN frameworks utilized are show similar statistical properties to those of the LES, without taking  physical evaluations into account during the GAN-training.

\begin{figure}[htb]
\centering
\includegraphics[scale=0.8]{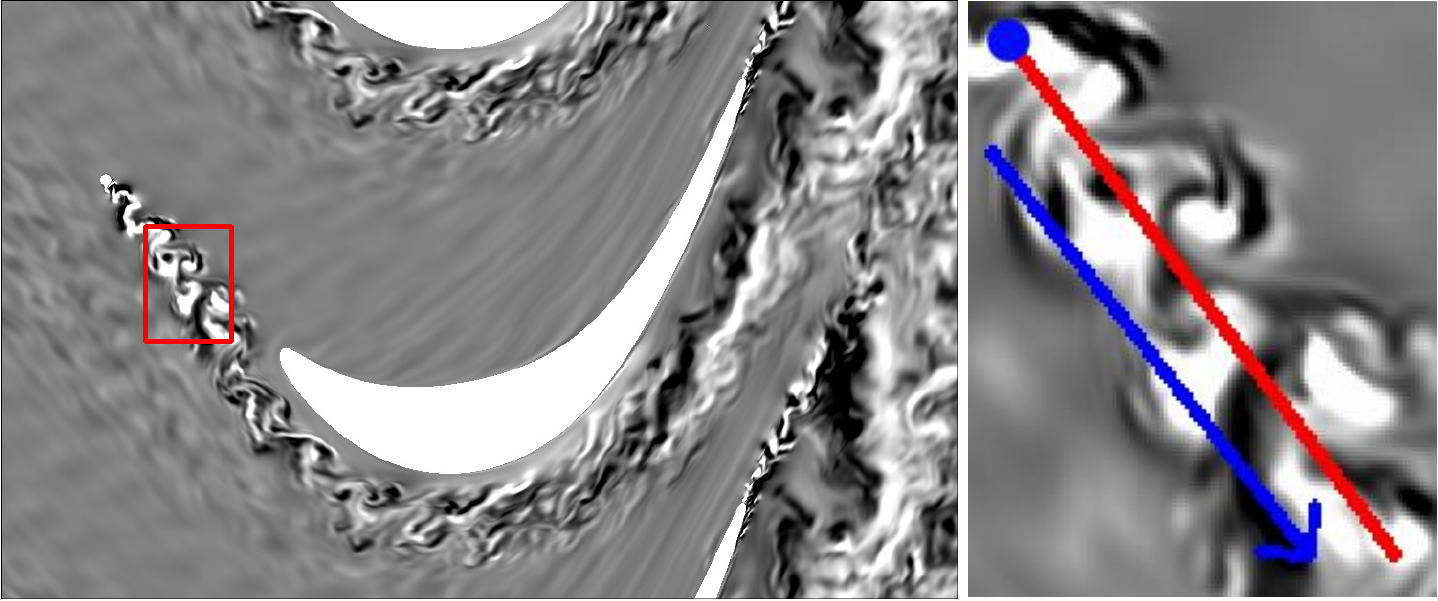}
\caption{Area investigated for the comparison of correlations of the turbulent flow field in case of the turbine stator under periodic wake impact. In the close-up on the right the directory of the vector $p$ is represented by the blue arrow and the pixel $\xi(\tau)$ by the blue point.}
\label{fig:evaluation_correlation_area}
\end{figure}

\begin{figure}[!h]
\centering
\includegraphics[width=0.7\textwidth]{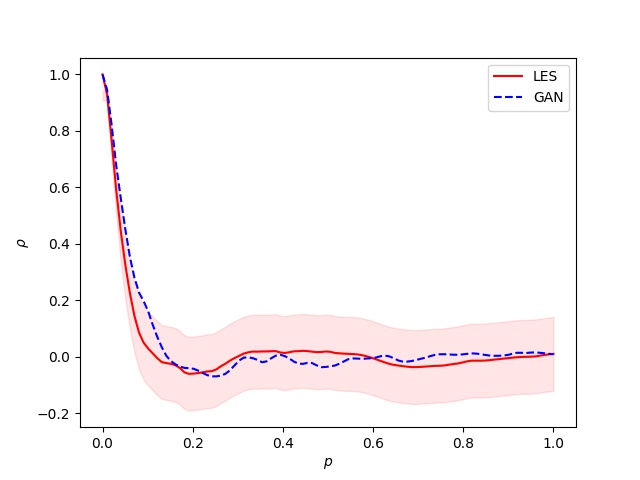}
\caption{Comparison of the pointwise correlation for the turbulent flow fields simulated by LES and generated by GAN along the $x$-axis. The red shaded area indicate the $95\%$ confidence intervals of the correlation for the respective curve.}
\label{fig:evaluation_result_correlation}
\end{figure}

\section{Conclusion and outlook}
\label{sec:conclusion}
We introduced generative adversarial networks as another way to model turbulence. In doing so, we showed that through generative learning it is possible to synthesize turbulence that matches the quality of LES images on a visual level while dramatically reducing computational time. Unlike previous work, we trained the GAN from scratch and only require a randomly sampled noise vector for the data production in the unconditional case. For training and inference of conditional GAN, we also need binary segmentation masks which can be created manually and do not necessarily need to be obtained by simulations. Using conditional GAN, we have found a solution for generating visually high-quality turbulence when solid objects as the rotation wake change position in space. Thus, we have provided a first approach to generalization with respect to spatial changes. Moreover, we have shown that the statistical properties of GAN-generated and LES flow agree excellently by investigating physics-based metrics. Finally, we have also demonstrated that generative learning of ergodic systems also works at the theoretical level.

So far, we have ignored the physics involved. Therefore, the next step is to feed the GAN with physical parameters so that turbulent flows can also be captured by the GAN in a physically correct manner and hence improve the results regarding the statistical properties even more.
Regarding the numerical experiments we will also pay attention to exploring and develping further appropriate evaluation methods. 
Having provided a first approach to generalization in terms of changes in turbulence space, in future work we will also consider how generalization can be realized in terms of geometries and further boundary conditions.

\section*{Acknowledgments}
\noindent C.D. and H.G. thank  Hayk Asatryan, Tobias Riedlinger and Matthias Rottmann for discussions and useful advice. The authors also thank Pascal Post for valuable hints for the literature research. We also thank two anonymous referees for hints that helped to improve this paper.

\appendix
\section{Training history of the GAN frameworks}
\label{sec:appendix}
The training progress of the experiments with the DCGAN discussed in \cref{ssec:results_karman} is described in \cref{fig:training_progress_karman}. Since we trained the GAN framework on images of size $ 512 \times 512$ we also got images of this size as output during the training. It can be observed that the synthesized images already show a quite good quality after $500$ epochs. However, on closer inspection, it is noticeable that the structures of the vortex street become finer with an increasing number of training epochs and that the color space is also captured much better after $2,000$. 
In \cref{fig:training_progress_lpt} the development of the synthesized images during the training is illustrated for the \texttt{pix2pixHD} whose results are discussed in \cref{ssec:results_lpt}. Similar to the DCGAN we can observe that the results improve significantly with increasing number of training epochs.

\begin{figure}[!h]
    \centering
    \includegraphics[width=\textwidth]{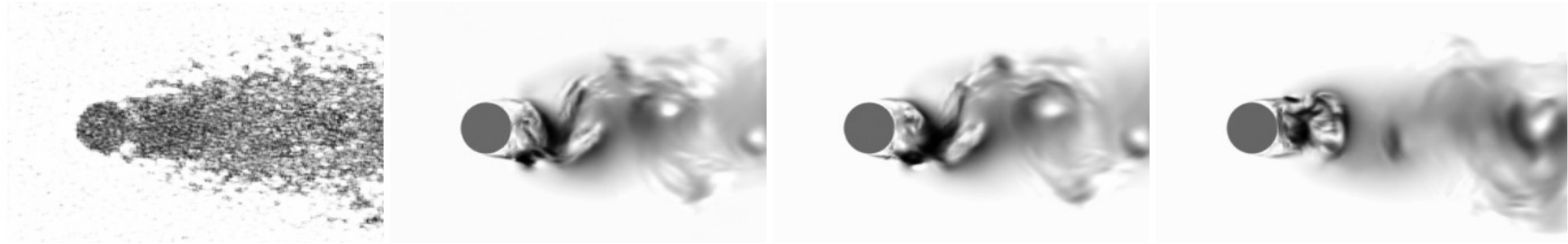}
    \caption{Development of training results after $1$, $500$, $1,500$ and $2,000$ epochs for the DCGAN.}
    \label{fig:training_progress_karman}
\end{figure}
\begin{figure}[!h]
    \centering
    \includegraphics[width=\textwidth]{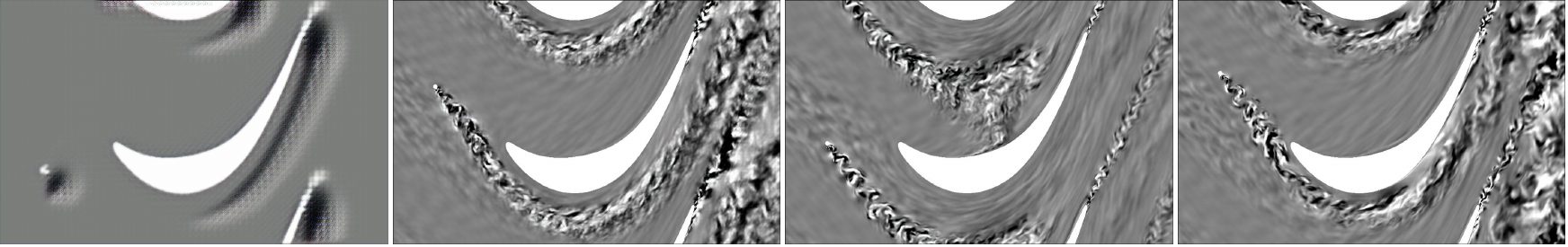}
    \caption{Development of training results after $1$, $50$, $150$ and $200$ epochs for the \texttt{pix2pixHD} framework.}
    \label{fig:training_progress_lpt}
\end{figure}

\section{Equivalence of Jensen-Shannon divergence and the norm of total variation}
\label{sec:appendix-B}

For our purpose, we only have to prove that convergence of probability measures $\nu_n\to\mu$ in the Jensen-Shannon divergence implies convergence in total variation.  In fact,

\begin{align*}
    \|\nu_n-\mu\|_\textup{TV}&\leq \left\|\nu_n-\frac{\mu+\nu_n}{2}\right\|_\textup{TV}+\left\|\mu-\frac{\mu+\nu_n}{2}\right\|_\textup{TV}\\
    &\leq \sqrt{2\mathfrak{d}_\textup{KL}\left(\nu_n\left\|\frac{\mu+\nu_n}{2}\right.\right)}+\sqrt{2\mathfrak{d}_\textup{KL}\left(\mu\left\|\frac{\mu+\nu_n}{2}\right.\right)}\\
    &\leq 2\sqrt{\mathfrak{d}_\textup{KL}\left(\nu_n\left\|\frac{\mu+\nu_n}{2}\right.\right)+\mathfrak{d}_\textup{KL}\left(\mu\left\|\frac{\mu+\nu_n}{2}\right.\right)}\\
    &=2\sqrt{\mathfrak{d}_\textup{JS}(\nu_n\|\mu)}\to 0.
\end{align*}
We used the triangle inequality in the first step, Pinsker's inequality \cite{van2014probability} in the second and the elementary inequality $\sqrt{x}+\sqrt{y}\leq \sqrt{2(x+y)}$ for real numbers $x,y\geq 0$ in the third step, which follows form the binomial formula and the inequality between the geometric and arithmetic mean.
For the converse statement, see e.g. \cite{lin1991divergence}.

\section{Evaluation results for larger grid of pixels}
\label{sec:appendix-C}
Evaluation results for the considered metrics in \cref{sec:physics_based_evaluation} for the area of $25\%$ of the pixel values after the wake.

\begin{figure*}[htb]
    \begin{subfigure}[b]{0.5\textwidth}
	    \includegraphics[width=\textwidth]{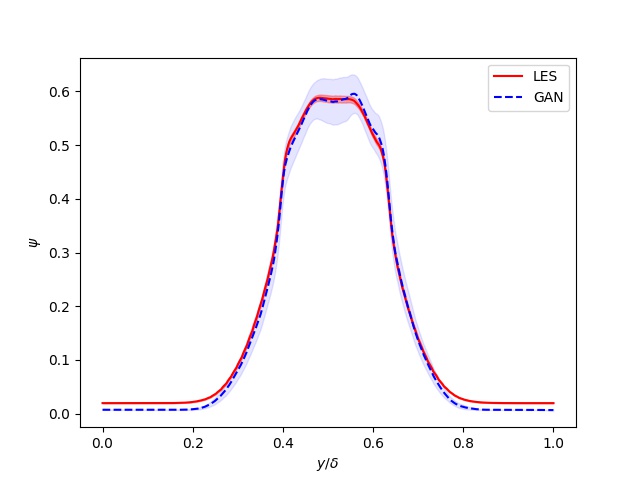}
	    \caption{Mean pixel values}
        \label{fig:eval_res_mean_larger}
    \end{subfigure}%
	\begin{subfigure}[b]{0.5\textwidth}
		\includegraphics[width=\textwidth]{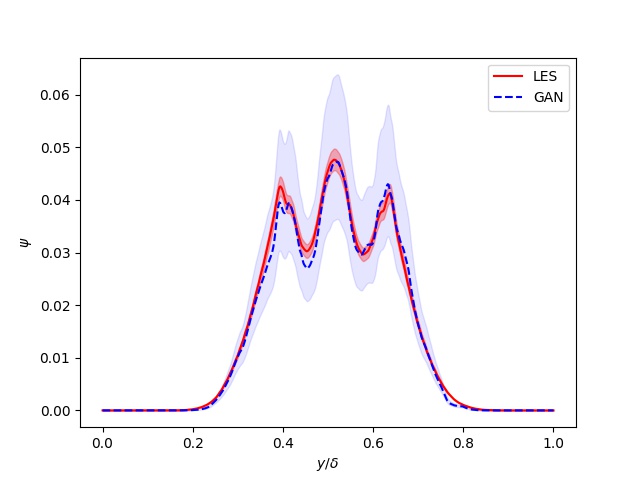}
		\caption{Statistical fluctuation of deviation from mean velocities}
		\label{fig:eval_res_var_larger}
	\end{subfigure}
\caption{Comparison of the mean pixel values (a) and the statistical fluctuation of the velocities (b) for the K\'arm\'an vortex street simulated by LES and generated by GAN along the $y$-axis. The blue and red shaded areas indicate the $95\%$ confidence intervals of the variance for the respective curves.}
\label{fig:evaluation_result_variance_mean_larger_region}
\end{figure*}

\bibliography{bibliography}

\begin{thebibliography}{10}

\bibitem{adams2003sobolev}
Robert~A Adams and John~JF Fournier.
\newblock {\em Sobolev spaces}.
\newblock Elsevier, 2003.

\bibitem{cnn_impact}
Saad Albawi, Tareq~Abed Mohammed, and Saad Al-Zawi.
\newblock Understanding of a convolutional neural network.
\newblock In {\em 2017 International Conference on Engineering and Technology
  (ICET)}, pages 1--6, 2017.

\bibitem{ergodicity_arbabi}
Hassan Arbabi and Igor Mezi{\'c}.
\newblock Ergodic theory, dynamic mode decomposition, and computation of
  spectral properties of the koopman operator.
\newblock {\em SIAM J. Appl. Dyn. Syst.}, 16:2096--2126, 2017.

\bibitem{wgan}
Martin Arjovsky, Soumith Chintala, and L{\'e}on Bottou.
\newblock {W}asserstein generative adversarial networks.
\newblock In Doina Precup and Yee~Whye Teh, editors, {\em Proceedings of the
  34th International Conference on Machine Learning}, volume~70 of {\em
  Proceedings of Machine Learning Research}, pages 214--223. PMLR, 06--11 Aug
  2017.

\bibitem{asatryan2020convenient}
Hayk Asatryan, Hanno Gottschalk, Marieke Lippert, and Matthias Rottmann.
\newblock A convenient infinite dimensional framework for generative
  adversarial learning.
\newblock {\em arXiv preprint arXiv:2011.12087}, 2020.

\bibitem{Beaudan:1995}
Patrick~Bruno Beaudan.
\newblock {\em Numerical experiments on the flow past a circular cylinder at
  sub-critical Reynolds number}.
\newblock PhD thesis, Stanford University, 1995.

\bibitem{Birkhoff_proof}
George~D. Birkhoff.
\newblock Proof of the ergodic theorem.
\newblock {\em Proceedings of the National Academy of Sciences},
  17(12):656--660, 1931.

\bibitem{wgan_gradient_clipping}
Xiangyi Chen, Steven~Z. Wu, and Mingyi Hong.
\newblock Understanding gradient clipping in private sgd: A geometric
  perspective.
\newblock In H.~Larochelle, M.~Ranzato, R.~Hadsell, M.~F. Balcan, and H.~Lin,
  editors, {\em Advances in Neural Information Processing Systems}, volume~33,
  pages 13773--13782. Curran Associates, Inc., 2020.

\bibitem{Cheung:2011}
Sai~Hung Cheung, Todd~A. Oliver, Ernesto~E. Prudencio, Serge Prudhomme, and
  Robert~D. Moser.
\newblock Bayesian uncertainty analysis with applications to turbulence
  modeling.
\newblock {\em Reliability Engineering \& System Safety}, 96(9):1137--1149,
  2011.
\newblock Quantification of Margins and Uncertainties.

\bibitem{odeint}
The~SciPy community.
\newblock Scipy documentation, 2008-2021.
\newblock Accessed: 04.12.2021.

\bibitem{Deng:2019}
Zhiwen Deng, Chuangxin He, Yingzheng Liu, and Kyung~Chun Kim.
\newblock Super-resolution reconstruction of turbulent velocity fields using a
  generative adversarial network-based artificial intelligence framework.
\newblock {\em Physics of Fluids}, 31(12):125111, 2019.

\bibitem{overfitting_1995}
Tom Dietterich.
\newblock Overfitting and undercomputing in machine learning.
\newblock {\em ACM computing surveys (CSUR)}, 27(3):326--327, 1995.

\bibitem{pooling}
Vincent Dumoulin and Francesco Visin.
\newblock A guide to convolution arithmetic for deep learning, 2018.

\bibitem{durbin_2018}
Paul~A. Durbin.
\newblock Some recent developments in turbulence closure modeling.
\newblock {\em Annual Review of Fluid Mechanics}, 50(1):77--103, 2018.

\bibitem{Edeling:2014a}
W.N. Edeling, P.~Cinnella, and R.P. Dwight.
\newblock Predictive rans simulations via bayesian model-scenario averaging.
\newblock {\em Journal of Computational Physics}, 275:65--91, 2014.

\bibitem{Edeling:2014b}
Wouter Edeling, Paola Cinnella, Richard Dwight, and Hester Bijl.
\newblock Bayesian estimates of parameter variability in the $k-\varepsilon$
  turbulence model.
\newblock {\em Journal of Computational Physics}, 258:73--94, 02 2014.

\bibitem{ergodicity_farkas}
Tanja Eisner, B{\'a}lint Farkas, Markus Haase, and Rainer Nagel.
\newblock {\em Operator Theoretic Aspects of Ergodic Theory}.
\newblock Springer International Publishing, Cham, 2015.

\bibitem{ferguson2017course}
Thomas~S Ferguson.
\newblock {\em A course in large sample theory}.
\newblock Routledge, 2017.

\bibitem{Ferziger:2008}
J.H. Ferziger and M.~Peri\'{c}.
\newblock {\em {"Computational Methods for Fluid Dynamics"}}.
\newblock \href{https://doi.org/10.1007/978-3-642-56026-2}{Springer}, Berlin,
  2008.

\bibitem{frisch1995turbulence}
Uriel Frisch and Andre{\u{\i}}~Nikolaevich Kolmogorov.
\newblock {\em Turbulence: the legacy of AN Kolmogorov}.
\newblock Cambridge university press, 1995.

\bibitem{Fukami:2019}
Kai Fukami, Koji Fukagata, and Kunihiko Taira.
\newblock Super-resolution reconstruction of turbulent flows with machine
  learning.
\newblock {\em Journal of Fluid Mechanics}, 870:106–120, May 2019.

\bibitem{Fukami:2020}
Kai Fukami, Koji Fukagata, and Kunihiko Taira.
\newblock Machine learning based spatio-temporal super resolution
  reconstruction of turbulent flows, 2020.

\bibitem{semSeg}
Alberto Garcia-Garcia, Sergio Orts-Escolano, Sergiu Oprea, Victor
  Villena-Martinez, and Jose Garcia-Rodriguez.
\newblock A review on deep learning techniques applied to semantic
  segmentation, 2017.

\bibitem{vanilla_gan}
Ian Goodfellow, Jean Pouget-Abadie, Mehdi Mirza, Bing Xu, David Warde-Farley,
  Sherjil Ozair, Aaron Courville, and Y.~Bengio.
\newblock Generative adversarial networks.
\newblock {\em Advances in Neural Information Processing Systems}, 3, 06 2014.

\bibitem{Goodfellow:2014}
Ian~J. Goodfellow, Jean Pouget-Abadie, Mehdi Mirza, Bing Xu, David
  Warde-Farley, Sherjil Ozair, Aaron Courville, and Yoshua Bengio.
\newblock Generative adversarial networks, 2014.

\bibitem{hagan_neural_networks}
Martin~T Hagan, Howard~B Demuth, and Mark Beale.
\newblock {\em Neural network design}.
\newblock PWS Publishing Co., 1997.

\bibitem{He:2018}
Chuangxin He, Yingzheng Liu, and Lian Gan.
\newblock A data assimilation model for turbulent flows using continuous
  adjoint formulation.
\newblock {\em Physics of Fluids}, 30:105108, 10 2018.

\bibitem{Hirsch:2007}
Charles Hirsch.
\newblock {\em {"Numerical Computation of Internal and External Flows: The
  Fundamentals of Computational Fluid Dynamics"}}.
\newblock
  \href{https://doi.org/10.1016/B978-0-7506-6594-0.X5037-1}{Butterworth-Heinemann
  }, 01 2007.

\bibitem{bce}
Yaoshiang Ho and Samuel Wookey.
\newblock The real-world-weight cross-entropy loss function: Modeling the costs
  of mislabeling.
\newblock {\em IEEE Access}, 8:4806--4813, 2020.

\bibitem{batch_norm}
Sergey Ioffe and Christian Szegedy.
\newblock Batch normalization: Accelerating deep network training by reducing
  internal covariate shift.
\newblock In Francis Bach and David Blei, editors, {\em Proceedings of the 32nd
  International Conference on Machine Learning}, volume~37 of {\em Proceedings
  of Machine Learning Research}, pages 448--456, Lille, France, 07--09 Jul
  2015. PMLR.

\bibitem{pix2pix}
Phillip Isola, Jun-Yan Zhu, Tinghui Zhou, and Alexei Efros.
\newblock Image-to-image translation with conditional adversarial networks.
\newblock pages 5967--5976, 07 2017.

\bibitem{Jiang:2020}
Chao Jiang, Junyi Mi, Shujin Laima, and Hui Li.
\newblock A novel algebraic stress model with machine-learning-assisted
  parameterization.
\newblock {\em Energies}, 13:258, 01 2020.

\bibitem{karras2019style}
Tero Karras, Samuli Laine, and Timo Aila.
\newblock A style-based generator architecture for generative adversarial
  networks.
\newblock In {\em Proceedings of the IEEE/CVF Conference on Computer Vision and
  Pattern Recognition}, pages 4401--4410, 2019.

\bibitem{kim2021unsupervised}
Hyojin Kim, Junhyuk Kim, Sungjin Won, and Changhoon Lee.
\newblock Unsupervised deep learning for super-resolution reconstruction of
  turbulence.
\newblock {\em Journal of Fluid Mechanics}, 910, 2021.

\bibitem{Kim:2020}
Junhyuk Kim and Changhoon Lee.
\newblock Deep unsupervised learning of turbulence for inflow generation at
  various reynolds numbers.
\newblock {\em Journal of Computational Physics}, 406:109216, 2020.

\bibitem{cnn_explanation}
Phil Kim.
\newblock {\em Convolutional Neural Network}, pages 121--147.
\newblock Apress, Berkeley, CA, 2017.

\bibitem{King:2017}
Ryan {King}, Peter {Graf}, and Michael {Chertkov}.
\newblock {Creating Turbulent Flow Realizations with Generative Adversarial
  Networks}.
\newblock In {\em APS Division of Fluid Dynamics Meeting Abstracts}, APS
  Meeting Abstracts, page A31.008, November 2017.

\bibitem{King:2018}
Ryan King, Oliver Hennigh, Arvind Mohan, and Michael Chertkov.
\newblock From deep to physics-informed learning of turbulence: Diagnostics,
  2018.

\bibitem{kingma2014adam}
Diederik~P Kingma and Jimmy Ba.
\newblock Adam: A method for stochastic optimization.
\newblock {\em arXiv preprint arXiv:1412.6980}, 2014.

\bibitem{Kolmogorov:1991}
A.~N. Kolmogorov.
\newblock {"The Local Structure of Turbulence in Incompressible Viscous Fluid
  for Very Large Reynolds Numbers"}.
\newblock {\em \href{https://doi.org/10.2307/51980}{Proceedings: Mathematical
  and Physical Sciences}}, 434(1890):9--13, 1991.

\bibitem{Kravchenko:2000}
Arthur~G. Kravchenko and Parviz Moin.
\newblock Numerical studies of flow over a circular cylinder at red=3900.
\newblock {\em Physics of Fluids}, 12(2):403--417, 2000.

\bibitem{Kuznetsov2020}
Nikolay~V. Kuznetsov, Timur~N. Mokaev, Olga~A. Kuznetsova, and Elena~V.
  Kudryashova.
\newblock The lorenz system: hidden boundary of practical stability and the
  lyapunov dimension.
\newblock {\em Nonlinear Dynamics}, 102:713--732, 2020.

\bibitem{lin1991divergence}
Jianhua Lin.
\newblock Divergence measures based on the shannon entropy.
\newblock {\em IEEE Transactions on Information theory}, 37(1):145--151, 1991.

\bibitem{dcgan_github}
Erik Linder-Nor\'en.
\newblock Pytorch-gan.
\newblock https://github.com/eriklindernoren/PyTorch-GAN.
\newblock Accessed: 12.11.2021.

\bibitem{Ling:2016}
Julia Ling, Andrew Kurzawski, and Jeremy Templeton.
\newblock Reynolds averaged turbulence modelling using deep neural networks
  with embedded invariance.
\newblock {\em Journal of Fluid Mechanics}, 807:155–166, 2016.

\bibitem{Liu:2020}
Bo~Liu, Jiupeng Tang, Haibo Huang, and Xi-Yun Lu.
\newblock Deep learning methods for super-resolution reconstruction of
  turbulent flows.
\newblock {\em Physics of Fluids}, 32(2):025105, 2020.

\bibitem{lorenz_1963}
Edward~N Lorenz.
\newblock Deterministic nonperiodic flow.
\newblock {\em Journal of atmospheric sciences}, 20(2):130--141, 1963.

\bibitem{fc_layer}
Wei Ma and Jun Lu.
\newblock An equivalence of fully connected layer and convolutional layer,
  2017.

\bibitem{cgan}
Mehdi Mirza and Simon Osindero.
\newblock Conditional generative adversarial nets, 2014.

\bibitem{Neumann_met}
J.~v. Neumann.
\newblock Proof of the quasi-ergodic hypothesis.
\newblock {\em Proceedings of the National Academy of Sciences}, 18(1):70--82,
  1932.

\bibitem{Norberg:1994}
C.~Norberg.
\newblock An experimental investigation of the flow around a circular cylinder:
  influence of aspect ratio.
\newblock {\em Journal of Fluid Mechanics}, 258:287–316, 1994.

\bibitem{activation_functions}
Chigozie Nwankpa, Winifred Ijomah, Anthony Gachagan, and Stephen Marshall.
\newblock Activation functions: Comparison of trends in practice and research
  for deep learning, 2018.

\bibitem{Ong:1996}
Lawrence Ong and James~M. Wallace.
\newblock The velocity field of the turbulent very near wake of a circular
  cylinder.
\newblock {\em Experiments in Fluids}, 20:441--453, 1996.

\bibitem{gan_survey}
Zhaoqing Pan, Weijie Yu, Xiaokai Yi, Asifullah Khan, Feng Yuan, and Yuhui
  Zheng.
\newblock Recent progress on generative adversarial networks (gans): A survey.
\newblock {\em IEEE Access}, 7:36322--36333, 2019.

\bibitem{Parish:2016}
Eric~J. Parish and Karthik Duraisamy.
\newblock A paradigm for data-driven predictive modeling using field inversion
  and machine learning.
\newblock {\em Journal of Computational Physics}, 305:758--774, 2016.

\bibitem{Parnaudeau:2008}
Philippe Parnaudeau, Johan Carlier, Dominique Heitz, and Eric Lamballais.
\newblock Experimental and numerical studies of the flow over a circular
  cylinder at reynolds number 3900.
\newblock {\em Physics of Fluids}, 20(8):085101, 2008.

\bibitem{pytorch2019}
Adam Paszke, Sam Gross, Francisco Massa, Adam Lerer, James Bradbury, Gregory
  Chanan, Trevor Killeen, Zeming Lin, Natalia Gimelshein, Luca Antiga, Alban
  Desmaison, Andreas Kopf, Edward Yang, Zachary DeVito, Martin Raison, Alykhan
  Tejani, Sasank Chilamkurthy, Benoit Steiner, Lu~Fang, Junjie Bai, and Soumith
  Chintala.
\newblock Pytorch: An imperative style, high-performance deep learning library.
\newblock In {\em Advances in Neural Information Processing Systems 32}, pages
  8024--8035. Curran Associates, Inc., 2019.

\bibitem{ergodicity_basics}
Ole Peters.
\newblock The ergodicity problem in economics.
\newblock {\em Nature Physics}, 15:1216--1221, 12 2019.

\bibitem{dcgan}
Alec Radford, Luke Metz, and Soumith Chintala.
\newblock Unsupervised representation learning with deep convolutional
  generative adversarial networks.
\newblock In Yoshua Bengio and Yann LeCun, editors, {\em 4th International
  Conference on Learning Representations, {ICLR} 2016, San Juan, Puerto Rico,
  May 2-4, 2016, Conference Track Proceedings}, 2016.

\bibitem{ruder2016overview}
Sebastian Ruder.
\newblock An overview of gradient descent optimization algorithms.
\newblock {\em arXiv preprint arXiv:1609.04747}, 2016.

\bibitem{backprop}
David~E. Rumelhart, Geoffrey~E. Hinton, and Ronald~J. Williams.
\newblock Learning representations by back-propagating errors.
\newblock {\em Nature}, 323:533--536, 1986.

\bibitem{shalev2014understanding}
Shai Shalev-Shwartz and Shai Ben-David.
\newblock {\em Understanding machine learning: From theory to algorithms}.
\newblock Cambridge university press, 2014.

\bibitem{Singh:2016a}
Anand Singh and Karthik Duraisamy.
\newblock Using field inversion to quantify functional errors in turbulence
  closures.
\newblock {\em Physics of Fluids}, 28:045110, 04 2016.

\bibitem{Singh:2016b}
Anand Singh, Shivaji Medida, and Karthik Duraisamy.
\newblock Machine-learning-augmented predictive modeling of turbulent separated
  flows over airfoils.
\newblock {\em AIAA Journal}, 55, 08 2016.

\bibitem{Subramanamiam:2020}
Akshay Subramaniam, Man~Long Wong, Raunak~D Borker, Sravya Nimmagadda, and
  Sanjiva~K Lele.
\newblock Turbulence enrichment using physics-informed generative adversarial
  networks, 2020.

\bibitem{lorenz_attractor_strange_attractore_proof}
Warwick Tucker.
\newblock The lorenz attractor exists.
\newblock {\em Comptes Rendus de l'Académie des Sciences - Series I -
  Mathematics}, 328(12):1197--1202, 1999.

\bibitem{van2014probability}
Ramon Van~Handel.
\newblock Probability in high dimension.
\newblock Technical report, PRINCETON UNIV NJ, 2014.

\bibitem{pix2pixHD}
Ting-Chun Wang, Ming-Yu Liu, Jun-Yan Zhu, Andrew Tao, Jan Kautz, and Bryan
  Catanzaro.
\newblock High-resolution image synthesis and semantic manipulation with
  conditional gans.
\newblock In {\em 2018 IEEE/CVF Conference on Computer Vision and Pattern
  Recognition}, pages 8798--8807, 2018.

\bibitem{wang2018esrgan}
Xintao Wang, Ke~Yu, Shixiang Wu, Jinjin Gu, Yihao Liu, Chao Dong, Yu~Qiao, and
  Chen Change~Loy.
\newblock Esrgan: Enhanced super-resolution generative adversarial networks.
\newblock In {\em Proceedings of the European conference on computer vision
  (ECCV) workshops}, 2018.

\bibitem{Weatheritt:2017}
J.~Weatheritt and R.D. Sandberg.
\newblock The development of algebraic stress models using a novel evolutionary
  algorithm.
\newblock {\em International Journal of Heat and Fluid Flow}, 68:298--318,
  2017.

\bibitem{Weatheritt:2016}
Jack Weatheritt and Richard Sandberg.
\newblock A novel evolutionary algorithm applied to algebraic modifications of
  the rans stress–strain relationship.
\newblock {\em Journal of Computational Physics}, 325:22--37, 2016.

\bibitem{Werhahn:2019}
Maximilian Werhahn, You Xie, Mengyu Chu, and Nils Thuerey.
\newblock A multi-pass gan for fluid flow super-resolution.
\newblock {\em Proceedings of the ACM on Computer Graphics and Interactive
  Techniques}, 2(2):1–21, Jul 2019.

\bibitem{winhart2020large}
Benjamin Winhart, Martin Sinkwitz, Andreas Schramm, Pascal Post, and Francesca
  di~Mare.
\newblock Large eddy simulation of periodic wake impact on boundary layer
  transition mechanisms on a highly loaded low-pressure turbine blade.
\newblock In {\em Turbo Expo: Power for Land, Sea, and Air}, volume 84102, page
  V02ET41A013. American Society of Mechanical Engineers, 2020.

\bibitem{Xie:2017}
You Xie, Eric Franz, Mengyu Chu, and Nils Thuerey.
\newblock Data-driven synthesis of smoke flows with cnn-based feature
  descriptors.
\newblock {\em ACM Transactions on Graphics}, 36(4):1–14, Jul 2017.

\bibitem{Yang:2020}
Muchen Yang and Zhixiang Xiao.
\newblock Improving the $k-\omega-\gamma-ar$ transition model by the field
  inversion and machine learning framework.
\newblock {\em Physics of Fluids}, 32, 06 2020.

\bibitem{yarotsky2017error}
Dmitry Yarotsky.
\newblock Error bounds for approximations with deep relu networks.
\newblock {\em Neural Networks}, 94:103--114, 2017.

\bibitem{Zhang:2018a}
Jincheng Zhang and Song Fu.
\newblock An efficient bayesian uncertainty quantification approach with
  application to $k-\omega-\gamma$ transition modeling.
\newblock {\em Computers \& Fluids}, 161:211--224, 2018.

\bibitem{Zhang:2018b}
Weiwei Zhang, Linyang Zhu, Jiaqing Kou, and Yilang Liu.
\newblock Machine learning methods for turbulence modeling in subsonic flows
  over airfoils, 06 2018.

\bibitem{Zhao:2020}
Yaomin Zhao, Harshal~D. Akolekar, Jack Weatheritt, Vittorio Michelassi, and
  Richard~D. Sandberg.
\newblock Rans turbulence model development using cfd-driven machine learning.
\newblock {\em Journal of Computational Physics}, 411:109413, 2020.

\bibitem{CycleGAN2017}
Jun-Yan Zhu, Taesung Park, Phillip Isola, and Alexei~A Efros.
\newblock Unpaired image-to-image translation using cycle-consistent
  adversarial networks.
\newblock In {\em Computer Vision (ICCV), 2017 IEEE International Conference
  on}, 2017.

\end{thebibliography}

\end{document}